\documentclass[a4paper,UKenglish,cleveref, autoref,thm-restate,authorcolumns]{lipics-v2019}

\title{Sensitive instances\texorpdfstring{\\}{ }of the Constraint Satisfaction Problem} 
\titlerunning{Sensitive instances of the Constraint Satisfaction Problem} 

\author{Libor Barto}
{Department of Algebra, Faculty of Mathematics and Physics, Charles University, Sokolovska 83, 186 00 Praha 8, Czech Republic}
{barto@karlin.mff.cuni.cz}
{https://orcid.org/0000-0002-8481-6458}
{Research partially supported by the European Research Council (ERC)
under the European Union's Horizon 2020 research and innovation programme (grant agreement No 771005), and by the Czech Science Foundation grant 18-20123S}

\author{Marcin Kozik}
{Theoretical Computer Science,
Faculty of Mathematics and Computer Science,
Jagiellonian University,
Kraków,
Poland}
{marcin.kozik@uj.edu.pl}
{https://orcid.org/0000-0002-1839-4824}
{Research partially supported by the National Science Centre Poland grants\\
{2014/13/B/ST6/01812} and 2014/14/A/ST6/00138}

\author{Johnson Tan}
{Department of Mathematics, University of Illinois, Urbana-Champaign, Urbana, IL USA, 61801}
{jgtan2@illinois.edu}
{https://orcid.org/0000-0002-5459-9712}
{}

\author{Matt Valeriote}
{Department of Mathematics and Statistics, McMaster University, 1280 Main Street West, Hamilton, Ontario, Canada L8S 4K1}
{matt@math.mcmaster.ca}
{https://orcid.org/0000-0001-6568-7526}
{Supported by the Natural Sciences and Engineering Research Council of Canada.}

\authorrunning{L. Barto, M. Kozik, J. Tan and M. Valeriote} 

\ccsdesc[500]{Theory of computation~Problems, reductions and completeness}
\ccsdesc[500]{Theory of computation~Complexity theory and logic}
\ccsdesc[300]{Theory of computation~Constraint and logic programming}

\keywords{Constraint satisfaction problem, bounded width, local consistency, near unanimity operation, loop lemma} 

\nolinenumbers 

\hideLIPIcs  

\usepackage{comment}
\usepackage{url}
\usepackage{todonotes}
\usepackage{amsmath,amssymb,amsthm}
\usepackage{latexsym}
\usepackage{hyperref}

\usepackage{tikz}
\usetikzlibrary{shapes,fit}

\newcommand{\vc}[1]{\mathbf{#1}}
\newcommand{\gp}[2]{#1^{\circ #2}}
\newcommand{\en}{\mathbb{N}}
\newcommand{\alg}[1]{{\mathbf #1}}
\newcommand{\csp}{{\mathsf {CSP}}}
\newcommand{\inst}[1]{{\mathcal #1}}
\newcommand{\var}[1]{{\mathcal #1}}
\newcommand{\constr}{{\mathcal C}}
\newcommand{\restr}[2]{#1_{|#2}}
\newcommand{\pat}[1]{\mathbb{#1}}
\newcommand{\asc}[1]{\mathcal{#1}}

\DeclareMathOperator{\proj}{proj}
\DeclareMathOperator{\dom}{dom}
\newcommand{\rel}[1]{{\mathbb #1}}

\newcommand{\V}[1]{\var V(\alg #1)}
\newcommand{\constraints}[1]{{\mathcal #1}}

\newtheorem{innercustomthm}{Theorem}
\newenvironment{customthm}[1]
  {\renewcommand\theinnercustomthm{#1}\innercustomthm}
  {\endinnercustomthm}

\begin{document}

\maketitle

\begin{abstract}
    We investigate the impact of modifying the constraining relations of a Constraint Satisfaction Problem ($\csp$) instance, with a fixed template,
    on the set of solutions of the instance.
    More precisely we investigate sensitive instances:
    an instance of the $\csp$ is called sensitive, if removing any tuple from any constraining relation invalidates some solution of the instance.
    Equivalently, one could require that every tuple from any one of its constraints extends to a solution of the instance.

    Clearly, any non-trivial template has instances which are not sensitive.
    Therefore we follow the direction proposed~%
    (in the context of strict width)
    by Feder and Vardi in \cite{feder-vardi} and require
    that only the instances produced by a local consistency checking algorithm are sensitive.
    In the language of the algebraic approach to the $\csp$
    we show that a finite idempotent algebra $\alg A$ has a $k+2$ variable near unanimity term operation if and only if any instance that results from running the $(k, k+1)$-consistency algorithm on an instance over $\alg A^2$ is sensitive.

    A version of our result, without idempotency but with
    the sensitivity condition holding in a variety of algebras, settles a question posed by G.\ Bergman about systems of projections of algebras that arise from some subalgebra of a finite product of algebras.

    Our results hold for infinite~%
    (albeit in the case of $\alg A$ idempotent) algebras as well
    and exhibit a surprising similarity to the  strict width~$k$ condition proposed by Feder and Vardi.
    Both conditions can be characterized by the existence of a near unanimity operation, but the arities of the operations differ by $1$.
\end{abstract}

\section{Introduction} \label{sec:intro}

One important algorithmic approach to deciding if a given instance of the Constraint Satisfaction Problem ($\csp$) has a solution is to first consider whether it has a consistent set of local solutions. Clearly, the absence of local solutions will rule out having any (global) solutions, but in general having local solutions does not guarantee the presence of a solution.  A major thrust of the recent research on the $\csp$  has focused on coming up with suitable notions of local consistency and then characterizing those $\csp$s for which local consistency implies outright consistency or some stronger property.
%
A good source for background material is the survey article~\cite{barto-krokhin-willard-how-to}.

Early results of Feder and Vardi \cite{feder-vardi} and also Jeavons, Cooper, and Cohen \cite{jeavons-cohen-cooper} establish that when a template (i.e., a relational structure) $\rel A$ has a special type of polymorphism, called a near unanimity operation, then not only will an instance of the $\csp$ over $\rel A$ that has a suitably consistent set of local solutions have a solution, but that any partial solution of it can always be extended to a solution.  The notion of local consistency that we investigate in this paper is related to that considered by these researchers but that, as we shall see, is weaker.

The following operations are central to our investigation.
\begin{definition}
An operation $n(x_1, \dots, x_{k+1})$ on a set $A$ of arity $k+1$ is called a {\em near unanimity operation on $A$} if it satisfies the equalities
\[
n(b,a,a, \dots, a) = n(a,b,a, \dots, a) = \dots = n(a,a, \dots, a, b) = a
\]
for all $a$, $b \in A$.
\end{definition}
  Near unanimity operations have played an important role in the development of universal algebra and first appeared in the 1970's in the work of Baker and Pixley \cite{BakerPixley} and Huhn \cite{Huhn1971}.  More recently they have been used in the study of the $\csp$  \cite{feder-vardi, jeavons-cohen-cooper} and related questions \cite{barto-nu, testability-focs}.  The main results of this paper can be expressed in terms of the $\csp$ and also in algebraic terms and we start by presenting them from both perspectives.
In the concluding section, Section~\ref{sec:conclusions}, a translation of parts of our results into a relational language is provided, along with some open problems.

\subsection{CSP viewpoint}

In their seminal paper, Feder and Vardi~\cite{feder-vardi} introduced
the notion of bounded width for the class of $\csp$ instances over  a finite template $\rel A$.  Their definition of bounded width was presented in terms of the logic programming language DATALOG but there is an equivalent formulation using local consistency algorithms, also given in~\cite{feder-vardi}.
Given a $\csp$ instance $\inst I$  and $k < l$, the $(k,l)$-consistency algorithm will produce a new instance having all $k$ variable constraints that can be inferred by considering $l$ variables  at a time of $\inst I$.  This algorithm rejects $\inst I$ if it produces an empty constraint. The class of $\csp$ instances over a finite template $\rel A$ will have width $(k,l)$ if the $(k,l)$-consistency algorithm rejects all instances from the class that do not have solutions, i.e., the $(k,l)$-consistency algorithm can be used to decide if a given instance from the class has a solution or not. The class has bounded width if it has width $(k,l)$ for some $k< l$.

A lot of effort, in the framework of the algebraic approach to the $\csp$, has gone  in to analyzing various properties of instances that are the outputs of these types of local consistency algorithms.
On one end of the spectrum of the research is a rather wide class of  templates of bounded width~\cite{barto-kozik-bw} and on the other a very restrictive class of templates having bounded strict width~\cite{feder-vardi}.

To be more precise, we now formally introduce instances of the CSP.

\begin{definition}
An  {\em instance} $\inst I$ of the $\csp$ is a pair $(V,\constraints C)$
where $V$ is a finite set of variables,
and $\constraints C$ is a set of constraints of the form $((x_1,\dotsc,x_n), R)$ where all $x_i$ are in $V$ and $R$ is an $n$-ary relation over (possibly infinite) sets $A_i$ associated to each variable $x_i$.

{\em A solution} of $\inst I$ is an evaluation $f$ of variables such that, for every $((x_1,\dotsc,x_n),R) \in \constraints C$ we have $(f(x_1),\dotsc,f(x_n))\in R$;
{\em a partial solution} is a partial function satisfying the same condition.

The \emph{$\csp$ over a relational structure $\rel A$}, written $\csp(\rel A)$, is  the class of $\csp$ instances whose constraint relations are from $\rel A$.
\end{definition}

\begin{example}\label{ex:colour}
For $k > 1$, the template associated with the graph $k$-colouring problem is the relational structure $\rel D_{k\text{colour}}$ that has universe $\{0,1, \dots, k-1\}$ and
a single relation $\ne_k = \{(x,y) \mid \text{$x, y < k$ and $x \ne y$}\}$.
The template associated with the HORN-3-SAT problem is the relational structure $\rel D_{\text{horn}}$ that has universe $\{0,1\}$ and two ternary relations $R_0$, $R_1$, where $R_i$ contains all the triples but $(1,1,i)$.
It is known that $\csp(\rel D_{\text{horn}})$ has width $(1,2)$, that $\csp(\rel D_{2\text{colour}})$ has width $(2,3)$, and that for
$k> 2$, $\csp(\rel D_{k\text{colour}})$  does not have bounded width (see \cite{barto-krokhin-willard-how-to}).
\end{example}

Instances produced by the $(k,l)$-consistency algorithm have uniformity and consistency properties that we highlight.
\begin{definition}
The $\csp$ instance $\inst I$ is {\em $k$-uniform}
if all of its constraints are $k$-ary and  every set of $k$ variables is constrained by a single constraint.

An instance is a \emph{$(k,l)$-instance} if it is $k$-uniform
and for every choice of a set $W$ of $l$ variables no additional information about the constraints can be derived by restricting the instance to the variables in $W$.
\end{definition}
This last, very important, property can be rephrased in the following way:
for every set $W\subseteq V$ of size $l$,
every tuple in every constraint of $\inst I_{|W}$ participates in a solution to $\inst I_{|W}$~%
(where $\inst I_{|W}$ is obtained from $\inst I$ by removing all the variables outside of $W$ and all the constraints that contain any such variables).

Consider the notion of strict width $k$~%
introduced by Feder and Vardi~\cite[Section 6.1.2]{feder-vardi}.
Let $\rel A$ be a template and let us assume, to avoid some technical subtleties, that every relation in $\rel A$ has arity at most $k$. The class $\csp(\rel A)$ has \emph{strict width $(k,l)$} if whenever the $(k, l)$-consistency algorithm does not reject an instance $\inst I$ from the class then ``it should be possible to obtain a solution by greedily assigning values to the variables one at a time while satisfying the inferred $k$-constraints.''
In other words,
if $\inst I$ is the result of applying the $(k,l)$-consistency algorithm to an instance of $\csp(\rel A)$, then any partial solution of $\inst I$ can be extended to a solution. The template $\rel A$ is said to have \emph{strict width $k$} if it has strict width $(k,l)$ for some $l > k$.

 A {\em polymorphism} of a template $\rel A$ is a function on $A$ that preserves all of the relations of $\rel A$.  Feder and Vardi prove the following.

\begin{theorem}[see Theorem 25, \cite{feder-vardi}] \label{thm:fv}
Let $k > 1$ and
    let $\rel A$ be a finite relational structure with relations of arity at most $k$.
The class $\csp(\rel A$) has strict width $k$
if and only if it has strict width $(k,k+1)$
if and only if $\rel A$ has a $(k+1)$-ary near unanimity operation as a polymorphism.
\end{theorem}

Using this Theorem we can conclude that $\csp(\rel D_{2\text{colour}})$ from Example~\ref{ex:colour} has strict width 2 since the ternary majority operation preserves the relation $\ne_2$.  In fact this operation preserves all binary relations over the set $\{0,1\}$. On the other hand, $\csp(\rel D_{\text{horn}})$ does not have strict width $k$ for any $k \geq 3$.

Following the algebraic approach to the $\csp$ we replace templates $\rel A$ with algebras $\alg A$.

\begin{definition}
An {\em algebra} $\alg A$ is a pair $(A, \mathcal F)$ where $A$ is a non-empty set, called the {\em universe} of $\alg A$ and $\mathcal F = (f_i \mid i \in I)$ is a set of finitary operations on $A$ called the {\em set of basic operations of $\alg A$.}  The function that assigns the arity of the operation $f_i$ to $i$ is called the {\em signature of $\alg A$}.  If $t(x_1, \dots, x_n)$ is a term in the signature of $\alg A$ then the interpretation of $t$ by $\alg A$ as an operation on $A$ is called a {\em term operation of $\alg A$} and is denoted by $t^{\alg A}$.

The \emph{$\csp$ over $\alg A$}, written $\csp(\alg A)$, is the class of $\csp$ instances whose constraint relations are amongst those relations over $A$ that are preserved by the operations of $\alg A$ (i.e., they are subuniverses of powers of $\alg A$).
\end{definition}
A number of important questions about the $\csp$ can be reduced to considering templates that have all of the singleton unary relations \cite{barto-krokhin-willard-how-to}; the algebraic counterpart to these types of templates are the  \emph{idempotent algebras}.
\begin{definition}
An operation $f:A^n \to A$ on a set $A$ is {\em idempotent} if $f(a, a, \dots, a) = a$ for all $a \in A$.  An algebra $\alg A$ is {\em idempotent} if all of its basic operations are.
\end{definition}
It follows that if $\alg A$ is idempotent then every term operation of $\alg A$ is an idempotent operation.
 As demonstrated in Example~\ref{Slupecki}, several of the results in this paper do not hold in the absence of idempotency.

The characterization of strict width in Theorem~\ref{thm:fv} has the following consequence in terms of algebras.
\begin{corollary} \label{cor:fv}
    Let $k > 1$ and
    let $\rel A$ be a finite relational structure with relations of arity at most $k$.
    Let $\alg A$ be the algebra with the same universe as $\rel A$ whose basic operations are exactly the polymorphisms of $\rel A$.  The following are equivalent:
    \begin{enumerate}
        \item $\alg A$  has a near unanimity term operation of arity $k+1$;
        \item in every $(k,k+1)$-instance over $\alg A$, every partial solution extends to a solution.
   \end{enumerate}
\end{corollary}

The implication ``1 implies 2'' in Corollary~\ref{cor:fv} remains valid for general algebras, not necessarily coming from finite relational structures with restricted arities of relations. However, the converse implication fails even if $\alg A$ is assumed to be finite and idempotent.

\begin{example}
Consider the rather trivial algebra $\alg A$ that has universe $\{0,1\}$ and no basic operations.  If $\inst I$ is a $(2,3)$-instance over $\alg A$ then since, as noted just after Theorem~\ref{thm:fv},  every binary relation over $\{0,1\}$ is invariant under the ternary majority operation on $\{0,1\}$ it follows that every partial solution of $\inst I$ can be extended to a solution.  Of course, $\alg A$ does not have a near unanimity term operation of any arity.
\end{example}

What this example demonstrates is that in general, for a fixed $k$, the $k$-ary constraint relations arising from an algebra do not capture that much of the structure of the algebra.  Example~\ref{Slupecki} provides further evidence for this.

Our first theorem shows that for finite idempotent algebras $\alg A$, by considering a slightly bigger set of $(k,k+1)$-instances, over $\csp(\alg A^2)$, rather than over $\csp(\alg A)$, we can detect the presence of a $(k+1)$-ary near unanimity term operation. Moreover, it is enough to consider only instances with $k+2$ variables.
We note that every $(k,k+1)$-instance over $\alg A$ can be easily encoded as a $(k,k+1)$-instance over $\alg A^2$.
\begin{theorem}\label{thm:swshort}
    Let $\alg A$ be a finite, idempotent algebra and $k>1$. The following are equivalent:
    \begin{enumerate}
        \item $\alg A$ (or equivalently $\alg A^2$) has a near unanimity term operation of arity $k+1$;
        \item in every $(k,k+1)$-instance over $\alg A^2$, every partial solution extends to a solution;
        \item in every $(k,k+1)$-instance over $\alg A^2$ \textbf{on $k+2$ variables}, every partial solution extends to a solution.
    \end{enumerate}
\end{theorem}
In Theorem ~\ref{thm:swfull} we extend our result to infinite idempotent algebras by working with local near unanimity term operations.

Going back the original definition of strict width: ``it should be possible to obtain a solution by greedily assigning values to the variables one at a time while satisfying the inferred $k$-constraints''
we note that the requirement that the assignment should be greedy is rather restrictive.
The main theorem of this paper investigates
an arguably more natural
concept where the assignment need not be greedy.
\begin{definition}
An instance of the $\csp$ is called \emph{sensitive}, if removing any tuple from any constraining relation invalidates some solution of the instance.
\end{definition}
In other words, an instance is sensitive if every tuple in every constraint of the instance extends to a solution.
For $(k, k+1)$-instances, being sensitive is equivalent to the instance being  a $(k,n)$-instance, where $n$ is the number of variables present in the instance.
We provide the following characterization.
\begin{theorem}\label{thm:sensshort}
    Let $\alg A$ be a finite, idempotent algebra
    and $k>1$. The following are equivalent:
    \begin{enumerate}
        \item $\alg A$~(or equivalently $\alg A^2$) has a near unanimity term operation of arity $k+2$;
        \item every $(k,k+1)$-instance over $\alg A^2$ is sensitive;
        \item every $(k,k+1)$-instance over $\alg A^2$  \textbf{on $k+2$ variables} is sensitive.
    \end{enumerate}
\end{theorem}
\noindent Exactly as in Theorem~\ref{thm:swshort} we can consider infinite algebras at the cost of using local near unanimity term operations (see Theorem~\ref{thm:sensfull}).

In conclusion we investigate a  natural property of instances motivated by the definition of strict width and
provide a characterization of this new condition in algebraic terms.
A surprising conclusion is that the new concept is, in fact, very close to the strict width concept, i.e., for a fixed $k$ one characterization is equivalent to a near unanimity operation of arity $k+1$ and the second of arity $k+2$.

\subsection{Algebraic viewpoint}
Our work  has as an antecedent the papers  of Baker and Pixley \cite{BakerPixley} and of Bergman \cite{Bergman1977} on algebras having near unanimity term operations. In these papers the authors considered subalgebras of products of algebras and systems of projections associated with them.  Baker and Pixley showed that in the presence of a near unanimity term operation, such a subalgebra is closely tied with its projections onto small sets of coordinates.

\begin{definition}
A \emph{variety of algebras} is a class of algebras of the same signature that is closed under taking homomorphic images, subalgebras, and direct products. For $\alg A$ an algebra, $\V A$ denotes the smallest variety that contains $\alg A$ and is called the \emph{variety generated by $\alg A$}. A variety $\var V$ has a near unanimity term of arity $k+1$ if there is some $(k+1)$-ary term in the signature of $\var V$ whose interpretation in each member of $\var V$ is a near unanimity operation.
\end{definition}

Here is one version of the Baker-Pixley Theorem:
\begin{theorem}[see Theorem 2.1 from \cite{BakerPixley}]\label{thm:BPold}  Let $\alg A$ be an algebra and $k > 1$.  The following are equivalent:
\begin{enumerate}
  \item $\alg A$ has a $(k+1)$-ary near unanimity term operation;
  \item\label{condition2} for every $r > k$ and every $\alg A_i \in \V A$, $1 \le i \le r$, every subalgebra $\alg R$ of $\prod_{i = 1}^r \alg A_i$ is
    \textbf{uniquely} determined by the projections of $R$ on all products $A_{i_1} \times \dots \times  A_{i_{k}}$ for $1 \le i_1 < i_2 < \dots < i_{k}\le r$;
  \item  the same as condition~\ref{condition2}, with $r$ set to $k+1$.
\end{enumerate}
\end{theorem}
In other words, an algebra has a $(k+1)$-ary near unanimity term operation if and only if
every subalgebra of a product of algebras from $\V A$ is uniquely determined by its system of $k$-fold projections into its factor algebras.
A natural question, extending the result above, was investigated by Bergman~\cite{Bergman1977}:
when does a given ``system of $k$-fold projections'' arise from a product algebra?

Note that such a system can be viewed as a $k$-uniform $\csp$ instance:
indeed, following the notation of Theorem~\ref{thm:BPold}, we can introduce a variable $x_i$ for each $i\leq r$
and a constraint $((x_{i_1},\dotsc,x_{i_k});\proj_{i_1,\dotsc,i_k} R)$ for each $1 \le i_1 < i_2 < \dots < i_{k}\le r$.
In this way the original relation $R$ consists of solutions of the created instance (but in general will not contain all of them).
In this particular instance, different variables can be evaluated in different algebras.
Note that the instance is sensitive, if and only if it ``arises from a product algebra'' in the sense investigated by Bergman.

We will say that $\inst I$ is a $\csp$ instance  {\em over the variety $\var V$ (denoted $\inst I \in \csp(\var V)$)} if all the constraining relations of $\inst I$ are algebras in $\var V$.
In the language of the $\csp$, Bergman proved the following:

\begin{theorem}[\cite{Bergman1977}]\label{thm:BergmanNew}
If $\var V$ is a variety that has a $(k+1)$-ary near unanimity term
then every $(k,k+1)$-instance over $\var V$ is sensitive.
\end{theorem}

In commentary that Bergman provided on his proof of this theorem he noted that a stronger conclusion could be drawn from it and he proved the following theorem.  We note that this theorem anticipates the results from~\cite{feder-vardi} and~\cite{jeavons-cohen-cooper} dealing with templates having near unanimity operations as polymorphisms.

\begin{theorem}[\cite{Bergman1977}]\label{thm:BP}
  Let $k > 1$ and   $\var V$ be a variety.  The following are equivalent:
  \begin{enumerate}
  \item $\var V$ has a $(k+1)$-ary near unanimity term;
  \item  any partial solution of a $(k,k+1)$-instance over $\var V$  extends to a solution.
  \end{enumerate}
\end{theorem}
\noindent We present a proof of this theorem in Appendix~\ref{sec:AlgebraProofs}.

Theorem~\ref{thm:BergmanNew} provides a partial answer to the question that Bergman posed in~\cite{Bergman1977}, namely that in the presence of a $(k+1)$-ary near unanimity term, a necessary and sufficient condition for a $k$-fold system of algebras to arise from a product algebra is that the associated $\csp$ instance is a $(k, k+1)$-instance.

In \cite{Bergman1977} Bergman asked whether the converse to Theorem~\ref{thm:BergmanNew} holds, namely,
that if all $(k,k+1)$-instances over a variety are sensitive,
must the variety have  a $(k+1)$-ary near unanimity term?
He provided examples that suggested that the answer is no, and we confirm this by proving that the condition is actually equivalent to the variety having a near unanimity term of arity $k+2$.
The main result of this paper, viewed from the algebraic perspective (but stated in terms of the $\csp$), is the following:

\begin{theorem}\label{thm:mainresult}
Let $k > 1$.  A variety $\var V$ has a $(k+2)$-ary near unanimity term if and only if each $(k, k+1)$-instance of the $\csp$ over $\var V$ is sensitive.
\end{theorem}
The ``if'' direction of this theorem is proved in Section~\ref{sect:gettingnu}, while a sketch of a proof of the ``only if'' direction can be found in Section~\ref{sect:sensitivity}~%
(the complete reasoning can be found in Appendix~\ref{sec:k+2}).
We note that a novel and significant feature of this result is that it does not assume any finiteness or idempotency of the algebras involved.

\subsection{Structure of the paper}
The paper is structured as follows.
In the next section we introduce  local near unanimity operations and state Theorem~\ref{thm:swshort} and Theorem~\ref{thm:sensshort} in their full power.
In Section~\ref{sect:gettingnu} we collect the proofs that establish the existence of (local) near unanimity operations.
Section~\ref{sect:newll} contains a proof of a new loop lemma, which can be of independent interest, and is necessary in the proof in Appendix~\ref{sec:k+2}.
In Section~\ref{sect:sensitivity} we provide a sketch of the proof showing that, in the presence of a near unanimity operation of arity $k+2$, the $(k,k+1)$-instances are sensitive.
A complete proof of this fact, which is our main contribution, can be found in Appendix~\ref{sec:k+2}.
Finally, Section~\ref{sec:conclusions} contains conclusions.

Appendix~\ref{sec:AlgebraProofs} and Appendix~\ref{app:ll} are provided for the convenience of the reader.
They prove facts required for the classification, but known before,
and facts which can be proved by minor adaptations of known reasoning.
Finally, Appendix~\ref{sec:k+2} contains, as already mentioned, the main technical contribution of the paper. 

\section{Details of the CSP viewpoint}
In order to state our results in their full strength, we need to define local near unanimity operations.
This special concept of local near unanimity operations is required,
when considering infinite algebras.

\begin{definition}
  Let $k > 1$.
  An algebra $\alg A$ has \emph{local near unanimity term operations of arity $k+1$} if for every finite subset $S$ of $A$ there is some $(k+1)$-ary term operation $n_S$ of $\alg A$ such that
  \[
  n_S(b,a, \dots, a, a) = n_S(a,b,a, \dots, a) = \dots = n_S(a,a, \dots, b,a) = n_S(a,a, \dots, a, b) = a.
  \]
 for all $a$, $b \in S$.
\end{definition}
It should be clear that, for finite algebras, having local near unanimity term  operations of arity $k+1$ and having a near unanimity term operation of arity $k+1$ are equivalent, but for arbitrary algebras they are not.   The following provides a  characterization of when an idempotent algebra has local near unanimity term operations of some given arity; it will be used in the proofs of Theorems~\ref{thm:swfull} and~\ref{thm:sensfull}. It is similar to Theorem~\ref{thm:BPold} and is proved in
Appendix~\ref{sec:AlgebraProofs}.

\begin{theorem}\label{thm:localBP}
Let $\alg A$ be an idempotent algebra and $k > 1$.  The following are equivalent:
\begin{enumerate}
\item $\alg A$ has local near unanimity term operations of arity $k+1$;
\item for every $r > k$, every  subalgebra of $\alg A^r$ is uniquely determined by its projections onto all $k$-element subsets of coordinates;
\item   every subalgebra of $\alg A^{k+1}$ is uniquely determined by its projections onto all $k$-element subsets of coordinates.
\end{enumerate}
\end{theorem}

We are ready to state Theorem~\ref{thm:swshort} in its full strength:
\begin{theorem}\label{thm:swfull}
    Let $\alg A$ be an idempotent algebra and $k>1$. The following are equivalent:
    \begin{enumerate}
        \item $\alg A$ (or equivalently $\alg A^2$) has local near unanimity term operations of arity $k+1$;
        \item in every $(k,k+1)$-instance over $\alg A^2$, every partial solution extends to a solution;
        \item in every $(k,k+1)$-instance over $\alg A^2$ \textbf{on $k+2$ variables}, every partial solution extends to a solution.
    \end{enumerate}
\end{theorem}
\begin{proof}
  Obviously condition 2 implies condition 3.
  A proof of condition 3 implying condition 1 can be found in Section~\ref{sect:gettingnu}.
  The implication from 1 to 2 is covered by Theorem~\ref{thm:BP}.
\end{proof}
Analogously, the main result of the paper, for idempotent algebras, and the full version of Theorem~\ref{thm:sensshort} states:
\begin{theorem}\label{thm:sensfull}
    Let $\alg A$ be an idempotent algebra
    and $k>1$. The following are equivalent:
    \begin{enumerate}
        \item $\alg A$~(or equivalently $\alg A^2$) has local near unanimity term operations of arity $k+2$;
        \item every $(k,k+1)$-instance over $\alg A^2$ is sensitive;
        \item every $(k,k+1)$-instance over $\alg A^2$  \textbf{on $k+2$ variables} is sensitive.
    \end{enumerate}
\end{theorem}
\begin{proof}
   Obviously condition 2 implies condition 3.
   For a proof that condition 3 implies condition 1 see Section~\ref{sect:gettingnu}.
   A sketch of the proof of the remaining implication can be found in Section~\ref{sect:sensitivity} (see
Appendix~\ref{sec:k+2} for a complete proof).
\end{proof}
\noindent The following examples show that in Theorems \ref{thm:localBP}, \ref{thm:swfull}, and \ref{thm:sensfull} the assumption of idempotency is necessary.

\begin{example}\label{Slupecki}
For $n > 2$, let $\alg S_n$ be the algebra with domain $[n] = \{1,2, \dots, n\}$ and with basic operations consisting of all unary operations  on $[n]$ and all non-surjective operations on $[n]$ of arbitrary arity.  The collection of such operations forms a finitely generated clone, called the S{\l}upecki clone.  Relevant details of these algebras can be found in \cite[Example 4.6]{kiss-prohle} and \cite{Szendrei2012}.  It can be shown that for $m <n$, the subuniverses of $\alg S_n^m$ consist of all $m$-ary relations $R_\theta$ over $[n]$ determined by a partition $\theta$ of $[m]$ by
\[
R_\theta = \{(a_1,  \dots, a_m)\mid \text{$a_i = a_j$ whenever $(i,j) \in \theta$}\}.
\]
These rather simple relations are preserved by any operation on $[n]$, in particular by any majority operation or more generally, by any near unanimity operation.

It follows from Theorem~\ref{thm:BP} that if $k > 1$ and $\inst I$ is a $(k,k+1)$-instance of $\csp(\alg S_{2k+1}^2)$ then any partial solution of $\inst I$ extends to a solution.  This also implies that $\inst I$ is sensitive.  Furthermore any subalgebra of $\alg S_{k+2}^{k+1}$ is determined by it projections onto all $k$-element sets of coordinates.
As noted in \cite[Example 4.6]{kiss-prohle}, for $n > 2$, $\alg S_n$ does not have a near unanimity term operation of any arity, since the algebra $\alg S_n^n$ has a quotient that is a 2-element essentially unary algebra.
\end{example}

\section{Constructing near unanimity operations}\label{sect:gettingnu}
In this section we collect the proofs providing,
under various assumptions, 
near unanimity or local near unanimity operations.
That is: the proofs of ``3 implies 1'' in Theorems~\ref{thm:swfull} and Theorem~\ref{thm:sensfull} 
as well as a proof of the ``if'' direction from Theorem~\ref{thm:mainresult}.

In the following proposition we construct instances over $\alg A^2$~%
(for some algebra $\alg A$). 
By a minor abuse of notation, we allow in such instances two kinds of variables: variables $x$ evaluated in $A$ and variables $y$ evaluated in $A^2$. 
The former kind should be formally considered as variables evaluated in $A^2$ where each constraint enforces that $x$
is sent to $\{(b,b)\mid b\in A\}$.

Moreover, dealing with $k$-uniform instances, we understand the condition ``every set of $k$ variables is constrained by a single constraint'' flexibly: in some cases we allow for more constraints with the same set of variables, as long as the relations are proper permutations so that every constraint imposes the same restriction.
\begin{proposition}\label{prop:sw}
   Let $k > 1$ and let $\alg A$  be an  algebra  such that, for every $(k,k+1)$-instance $\inst I$ over $\alg A^2$ on $k+2$ variables every partial solution of $\inst I$ extends to a solution.
   Then each subalgebra of $\alg A^{k+1}$ is determined by its $k$-ary projections.
\end{proposition}
\begin{proof}
Let $\alg R\leq \alg A^{k+1}$ and we will show that
it is determined by the system of projections $\proj_I(R)$ as $I$ ranges over all $k$ elements subsets of coordinates.
Using $\alg R$  we define the following  instance $\inst I$ of $\csp(\alg A^2)$. The variables of $\inst I$ will be the set $\{x_1, x_2, \dots, x_{k+1}, y_{12}\}$ and the domain  of each $x_i$ is $A$, while the domain of $y_{12}$ is $A^2$.

For $U \subseteq \{x_1, \dots, x_{k+1}\}$ of size $k$, let $C_U$ be the constraint with scope $U$ and constraint relation $R_U = \proj_U(R)$.  
For $U$ a $(k-1)$-element subset of $\{x_1, \dots, x_{k+1}\}$, let $C_{U\cup \{y_{12}\}}$ be the constraint  with scope $U \cup \{y_{12}\}$ and constraint relation $R_{U\cup\{y_{12}\}}$ that consists of all tuples $(b_v \mid v \in U\cup \{y_{12}\})$ such that there is some $(a_1, \dots, a_{k+1}) \in R$ with $b_v = a_i$ if $v = x_i$ and with $b_{y_{12}} = (a_1, a_2)$.

The instance $\inst I$ is $k$-uniform and we will show that it is sensitive.
Indeed every tuple in every constraining relation originates in some tuple $\vc{b}\in R$.
Setting $x_i\mapsto b_i$ and $y_{12}\mapsto (b_1,b_2)$ defines a solution that extends such a tuple.

In particular $\inst I$ is a $(k, k+1)$-instance over $\alg A^2$ with $k+2$ variables and so any partial solution of it can be extended to a solution.  Let $\vc{b} \in A^{k+1}$ such that $\proj_I(\vc{b}) \in \proj_I(R)$ for all $k$ element subsets $I$ of $[k+1]$.  Then $\vc{b}$ is a partial solution of $\inst I$ over the variables $\{x_1, \dots, x_{k+1}\}$ and thus there is some extension of it to the variable $y_{12}$ that produces a solution of $\inst I$.  But there is only one consistent way to extend $\vc{b}$ to $y_{12}$ namely by setting $y_{12}$ to the value $(b_1, b_2)$.  By considering the constraint with scope $\{x_3, \dots, x_{k+1}, y_{12}\}$ it follows that $\vc{b} \in R$, as required.
\end{proof}
Now we are ready to prove the first implication tackled in this section: 3 implies 1 in Theorem~\ref{thm:swfull}.

\begin{proof}
    [Proof of ``3 implies 1'' in Theorem~\ref{thm:swfull}]
    By Theorem~\ref{thm:localBP} it suffices to show that
    each subalgebra of $\alg A^{k+1}$ is determined by its $k$-ary projections. 
    Fortunately, Proposition~\ref{prop:sw} provides just that.
\end{proof}

We move on to  proofs of ``3 implies 1'' in  Theorem~\ref{thm:sensfull} and the ``if''  direction of Theorem~\ref{thm:mainresult}.
Similarly, as in the theorem just proved, we start with a proposition.
\begin{proposition}\label{prop:sens}
   Let $k > 1$ and let $\alg A$  be an  algebra  such that every $(k,k+1)$-instance $\inst I$ over $\alg A^2$ on $k+2$ variables is sensitive.
   Then each subalgebra of $\alg A^{k+2}$ is determined by its $(k+1)$-ary projections.
\end{proposition}
\begin{proof}
  We will show that if   $\alg R$ is a subalgebra of $\alg A^{k+2}$ then $R =  R^*$ where
  \[
   R^* = \{a\in A^{k+2} \mid \proj_I(a) \in \proj_I(R) \text{ whenever } |I|=k+1\}.
  \]
  In other words, we will show that the subalgebra $\alg R$ is determined by its  projections into all $(k+1)$-element sets of coordinates.

    We will use $ R$ and $R^*$ from the previous paragraph to construct a $(k,k+2)$-instance 
    $\inst I = (V, \constr)$
    with $V = \{x_5, \dots, x_{k+2}, y_{12}, y_{34}, y_{13}, y_{24}\}$
    where each $x_i$ is evaluated in $A$ while all the $y$'s are evaluated in $A^2$.
   
    The set of constraints is more complicated.
    There is a {\em special constraint} on a {\em special variable set}
    $((y_{12}, y_{34}, x_5, \dots, x_{k+2}), C)$ where
    \[
    C= \{((a_1, a_2), (a_3, a_4), a_5, \dots, a_{k+2}) \mid (a_1, \dots, a_{k+2}) \in R^*\}.
    \]
    The remaining constraints are defined using the relation $R$.
    For each set of variables $S = \{v_1, \dots, v_k\} \subseteq V$~%
    (which is different than the set for the special constraint)
    we define a constraint $((v_1,\dotsc,v_k), D_S)$
    with $(b_1, \dots, b_k) \in D_S$ 
    if and only if there exists a tuple $(a_1,\dots,a_{k+2})\in R$ such that:
    \begin{itemize}
        \item if $v_i$ is $x_j$ then $b_i = a_j$, and
        \item if $v_i$ is $y_{lm}$ then $b_i=(a_l,a_m)$.
    \end{itemize}
    Note that the instance $\inst I$ is $k$-uniform.

   \begin{claim}
     $\inst I$ is a $(k, k+1)$-instance.
   \end{claim}

   Let $S \subseteq V$ be a set of size $k$.
   If $S$ is not the special variable set,
   then every tuple in the relation constraining $S$ originates in some $(b_1,\dotsc,b_{k+2})\in R$ and, as in Proposition~\ref{prop:sw}, 
   sending $x_i\mapsto b_i$ and $y_{lm}\mapsto (b_l,b_m)$ defines a solution that extends such a tuple.
   We immediately conclude, that the potential failure of the $(k,k+1)$ condition must involve the special constraint.

   Thus $S = \{y_{12},y_{34},x_5, \dots, x_{k+2}\}$ and if $\vc{b}$ is a tuple from the special constraint $C$ then there is some $(a_1, \dots, a_{k+2}) \in R^*$ with
   \[
   \vc{b} = ((a_1, a_2), (a_3, a_4), a_5, \dots, a_{k+2}).
   \]
   The extra variable that we want to extend the tuple $\vc{b}$ to is either $y_{13}$ or $y_{24}$.
   Both cases are similar and we will only work through the details when it is $y_{13}$. 
   In this case, assigning the value $(a_1, a_3)$ to the variable $y_{13}$ will produce an extension $\vc{b}'$ of $\vc{b}$ to a tuple over $S \cup\{y_{13}\}$ that is consistent with all constraints of $\inst I$ whose scopes are subsets of $\{y_{12},y_{34},x_5, \dots, x_{k+2}, y_{13}\}$.

   To see this, consider a $k$ element subset $S'$ of 
   $\{y_{12},y_{34},x_5, \dots, x_{k+2}, y_{13}\}$
   that excludes some variable $x_j$.  Then, by the definition of $R^*$ there exists some tuple of the form $(a_1, a_2, \dots, a_{j-1}, a_{j}', a_{j+1}, \dots, a_{k+2}) \in R$.  This tuple from $R$ can be used to witness that the restriction of $\vc{b}'$ to $S'$ satisfies the constraint $D_{S'}$ since the scope of this constraint does not include the variable $x_j$.

   Suppose that $S'$ is a $k$ element subset of 
   $\{y_{12},y_{34},x_5, \dots, x_{k+2}, y_{13}\}$
   that excludes $y_{12}$. 
   By the definition of $R^*$ there is some tuple of the form $(a_1, a_2', a_3,\dots, a_{k+2}) \in R$.  Using this tuple it follows that the restriction of $\vc{b}'$ to $S'$ satisfies the constraint $D_{S'}$.  This is because neither of the variables $y_{12}$ and $y_{24}$ are in $S'$ and so the value $a_2' \in A_2$ does not matter.  A similar argument works when $S'$ is assumed to exclude $y_{34}$ and the claim is proved.

   \medskip

   Since  $\inst I$ is a $(k, k+1)$-instance over $\alg A^2$ and it has $k+2$ variables 
   then by assumption,  $\inst I$ is sensitive.  We can use this to show that $R^* \subseteq R$ to complete the proof of this proposition.  Let $(a_1, \dots, a_{k+2}) \in R^*$ and consider the associated tuple $\vc{b} = ((a_1, a_2), (a_3, a_4), a_5, \dots, a_{k+2}) \in C$.
   Since $\inst I$ is sensitive then this $k$-tuple can be extended to a solution $\vc{b}'$ of $\inst I$.  Using any constraints of $\inst I$ whose scopes include combinations of $y_{12}$ or $y_{34}$ with $y_{13}$ or $y_{24}$ it follows that the value of $\vc{b}'$ on the variables $y_{13}$ and $y_{24}$ are $(a_1, a_3)$ and $(a_2, a_4)$ respectively.  Then considering the restriction of $\vc{b}'$ to $S = \{x_5, \dots, x_{k+2}, y_{13}, y_{24}\}$ it follows that $(a_1, \dots, a_{k+2}) \in R$ since this restriction lies in the constraint relation $D_S$.
\end{proof}

\noindent We are in a position to provide the two final proofs in this section.
\begin{proof}
    [Proof of ``3 implies 1'' in Theorem~\ref{thm:sensfull}]
    By Theorem~\ref{thm:localBP} it suffices to show that
    each subalgebra of $\alg A^{k+2}$ is determined by its $(k+1)$-ary projections. 
    Fortunately Propositions~\ref{prop:sens} provides just that.
\end{proof}
\begin{proof}
    [Proof of the ``if'' direction in Theorem~\ref{thm:mainresult}]
    For this direction we apply Proposition~\ref{prop:sens} to a special member of $\var V$, namely the $\var V$-free algebra freely generated by $\mathbf x$ and $\mathbf y$, which we will denote by $\alg F$.  Up to isomorphism, this algebra is unique and its defining property is that $\alg F \in \var V$ and for any algebra $\alg A \in \var V$, any map $f: \{\mathbf x,\mathbf y\} \to A$ extends uniquely to a homomorphism from $\alg F$ to $\alg A$.  Consequently, for any two terms $s(x,y)$ and $t(x,y)$ in the signature of $\var V$ if $s^{\alg F}({\mathbf x},{\mathbf y}) = t^{\alg F}({\mathbf x},{\mathbf y})$ then the equation $s(x,y) \approx t(x,y)$ holds in $\var V$.
   
    Let $\alg R$ be the subalgebra of $\alg F^{k+2}$ generated by the tuples $({\mathbf y},{\mathbf x},{\mathbf x},\dotsc,{\mathbf x})$,
    $({\mathbf x},{\mathbf y},{\mathbf x},\dotsc,{\mathbf x})$, \dots,
    $({\mathbf x},\dotsc,{\mathbf x},{\mathbf y})$.
    By Proposition~\ref{prop:sens}, the algebra $\alg R$ is determined by its $(k+1)$-ary projections and so the constant tuple $({\mathbf x},\dotsc,{\mathbf x})$
    belongs to $R$. 
    The term generating this tuple from the given generators of $\alg R$ defines the required $(k+2)$-ary near unanimity operation.
\end{proof}

\section{New loop lemmata}\label{sect:newll}
A \emph{loop lemma} is a theorem stating that a binary relation satisfying certain structural and algebraic requirements necessarily contains a \emph{loop} -- a pair $(a,a)$.
In this section we provide two new loop lemmata, Theorem~\ref{thm:loop_cycle} and Theorem~\ref{thm:loop_inv}, which generalize an ``infinite loop lemma'' of Ol\v s\'ak~\cite{O17} and may be of independent interest. Theorem~\ref{thm:loop_inv} is a crucial tool for the proof sketched  in Section~\ref{sect:sensitivity} and presented in Appendix~\ref{sec:k+2}.

The algebraic assumptions in the new loop lemmata concern absorption, a concept that has proven to be useful in the algebraic theory of $\csp$s and in universal algebra~\cite{AbsorptionSurvey}. We adjust the standard definition to our specific purposes.
We begin with a very elementary definition.
\begin{definition}
    Let $R$ and $S$ be sets.
    We call a tuple $(a_1,\dotsc,a_n)$ a
    {\em one-$S$-in-$R$ tuple}
    if for exactly one $i$ we have $a_i\in S$ and all the other $a_i$'s are in $R$.
\end{definition}
Next we proceed to define a relaxation of the standard absorbing notion.
We follow a standard notation, silently extending operations of an algebra to powers~%
(by computing them coordinate-wise).
\begin{definition}
    Let $\alg A$ be an algebra, $\alg R\leq\alg A^k$ and $S \subseteq A^k$.
    We say that $R$ locally $n$-absorbs $S$ if,
    for  every finite set $\mathcal C$ of one-$S$-in-$R$ tuples of length $n$, there is a term operation $t$ of $\alg A$
    such that $t(\vc{a^1},\dotsc,\vc{a^n})\in R$ whenever $(\vc{a^1},\dotsc,\vc{a^n})\in\mathcal C$.
    We will say that $R$ {\em locally absorbs} $S$, if $R$ locally $n$-absorbs $S$ for some $n$.
\end{definition}

Absorption, even in this  form, is stable under various constructions. The following lemma lists some of them  and we leave it without a proof~%
(the reasoning is identical to the one in e.g. Proposition 2 in~\cite{AbsorptionSurvey}).

\begin{lemma}~\label{lem:abs_pp}
Let $\alg A$ be an algebra and $\alg R\leq \alg A^2$ such that $R$ locally $n$-absorbs $S$.
Then $R^{-1}$ locally $n$-absorbs $S^{-1}$;
and $R\circ R$ locally $n$-absorbs $S\circ S$,
and $R\circ R\circ R$ locally $n$-absorbs $S\circ S\circ S$ etc.
\end{lemma}
Let us prove a first basic property of local absorption.
\begin{lemma} \label{lem:aux_loops}
Let $\alg A$ be an idempotent algebra and $\alg R\leq \alg A^2$ such that $R$ locally $n$-absorbs $S$.
Let $(a_1, \ldots, a_n)$ and $(b_1, \ldots, b_n)$ be directed walks in $R$, and let $(a_i,b_i) \in S$ for each $i$ (see Figure~\ref{fig:Lemma20}).
Then there exists a directed walk from $a_1$ to $b_n$ of length $n$ in $R$.
\end{lemma}

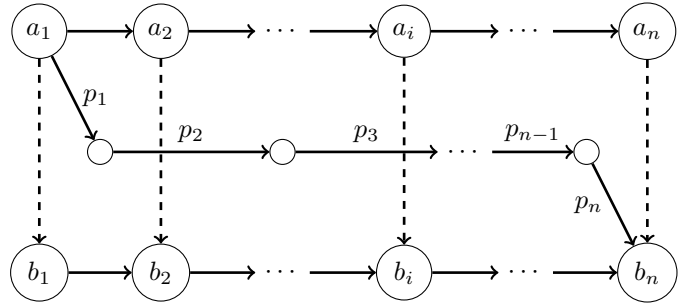
\begin{figure}
\centering
\begin{tikzpicture}[scale=0.8]
    \node[shape=circle,draw=black] (a_1) at (0,0) {$a_1$};
    \node[shape=circle,draw=black] (a_2) at (2,0) {$a_2$};
    \node[draw=none] (blank1) at (4,0) {$\cdots$};
    \node[shape=circle,draw=black] (a_i) at (6,0) {$a_i$};
    \node[draw=none] (blank2) at (8,0) {$\cdots$};
    \node[shape=circle,draw=black] (a_r) at (10,0) {$a_n$};

    \node[shape=circle,draw=black] (b_1) at (0,-4) {$b_1$};
    \node[shape=circle,draw=black] (b_2) at (2,-4) {$b_2$};
    \node[draw=none] (blank3) at (4,-4) {$\cdots$};
    \node[shape=circle,draw=black] (b_i) at (6,-4) {$b_i$};
    \node[draw=none] (blank4) at (8,-4) {$\cdots$};
    \node[shape=circle,draw=black] (b_r) at (10,-4) {$b_n$};

    \node[shape=circle,draw=black] (blankpath1) at (1,-2) {};
    \node[shape=circle,draw=black] (blankpath2) at (4,-2) {};
    \node[draw=none] (blankpath3) at (7,-2) {$\cdots$};
    \node[shape=circle,draw=black] (blankpath4) at (9,-2) {};

    \path [->,line width=1pt] (a_1) edge node[left] {} (a_2);
    \path [->,line width=1pt] (a_2) edge node[left] {} (blank1);
    \path [->,line width=1pt] (blank1) edge node[left] {} (a_i);
    \path [->,line width=1pt] (a_i) edge node[left] {} (blank2);
    \path [->,line width=1pt] (blank2) edge node[left] {} (a_r);

    \path [->,line width=1pt] (b_1) edge node[left] {} (b_2);
    \path [->,line width=1pt] (b_2) edge node[left] {} (blank3);
    \path [->,line width=1pt] (blank3) edge node[left] {} (b_i);
    \path [->,line width=1pt] (b_i) edge node[left] {} (blank4);
    \path [->,line width=1pt] (blank4) edge node[left] {} (b_r);

    \path [->,dashed,line width=1pt] (a_1) edge node[left] {} (b_1);
    \path [->,dashed,line width=1pt] (a_2) edge node[left] {} (b_2);
    \path [->,dashed,line width=1pt] (a_i) edge node[left] {} (b_i);
    \path [->,dashed,line width=1pt] (a_r) edge node[left] {} (b_r);

    \path [->,line width=1pt] (a_1) edge node[right] {$p_1$} (blankpath1);
    \path [->,line width=1pt] (blankpath1) edge node[left,above] {$p_2$} (blankpath2);
    \path [->,line width=1pt] (blankpath2) edge node[left,above] {$p_3$} (blankpath3);
    \path [->,line width=1pt] (blankpath3) edge node[left,above] {$p_{n-1}$} (blankpath4);
    \path [->,line width=1pt] (blankpath4) edge node[left] {$p_n$} (b_r);

\end{tikzpicture}

\caption{Solid arrows represent tuples from R and dashed arrows represent tuples from S.} \label{fig:Lemma20}

\end{figure}

\begin{proof}
We will show that there is a term operation $t$ of the algebra $\alg A$ such that the following $(n+1)$-tuple of elements of $A$ is a walk of length $n$ in $R$ from $a_1$ to $b_n$.
\begin{align*}
    ( a_1 =  &t(a_1, a_1, a_1, \ldots, a_1), \\
        &t(b_1, a_2, a_2, \ldots, a_2), \\
        &t(b_2, b_2, a_3, \ldots, a_3), \\
    &\vdots \\
        &t(b_{n-1}, b_{n-1}, \ldots, b_{n-1}, a_{n}), \\
    b_n =  &t(b_n, b_n,b_n, \ldots, b_n) ).
\end{align*}
In order to choose a proper $t$ we apply the definition of local absorption to the set of $(n+1)$
one-$S$-in-$R$ tuples corresponding to the steps in the path.
\end{proof}

The loop lemma of Ol\v s\'ak concerns symmetric relations absorbing the equality relation $\{(a,a) \mid a\in A\}$, which is denoted $=_A$.
The original result, stated in a slightly different language, does not cover the case of local absorption.
However, a typographical modification of a proof mentioned in~\cite{O17} shows that the theorem holds.
For completeness sake, we present this proof in
Appendix~\ref{app:ll}.
\begin{theorem}[\cite{O17}] \label{thm:olsak_loop}
    Let $\alg A$ be an \textbf{idempotent} algebra and
    $\alg R\leq \alg A^2$ be nonempty and symmetric.
    If $R$ locally absorbs $=_A$, then $R$ contains a loop.
\end{theorem}

In order to apply this theorem in the case of  sensitive instances, we need to generalize it.
In the following two theorems we will gradually relax the requirement that $R$ is symmetric.
In the first step, we substitute it with a condition requiring a closed, directed walk in the graph~%
(i.e., a sequence of possibly repeating vertices, with consecutive vertices connected by forward edges and the first and last vertex identical).
Recall that $R^{-1}$ is the inverse relation to $R$ and let us denote by $\gp{R}{l}$ the $l$-fold relational composition of $R$ with itself.

\begin{theorem} \label{thm:loop_cycle}
    Let $\alg A$ be an \textbf{idempotent} algebra and $\alg R \leq \alg A^2$ contain a directed closed walk.
    If $R$ locally absorbs $=_A$, then $R$ contains a loop.
\end{theorem}

\begin{proof}
    Let $n$ denote the arity of the absorbing operations.
    The proof is by induction on $l \geq 0$, where $l$ is a number such that there exists a directed closed walk from $a_1$ to $a_1$ of length $2^l$.

    We start by verifying that such an $l$ exists.
    Take a directed walk $(a_1, \ldots, a_{k-1}, a_k = a_1)$ in $R$. We may assume that its length $k$ is at least $n$, since we can, if necessary, traverse the walk multiple times.
    An application of Lemma~\ref{lem:aux_loops} to the relations $R, =_A$ and tuples $(a_1, \dots, a_n)$, $(a_1, \dots, a_n)$ gives us
    a directed walk from $a_1$ to $a_n$ of length $n$. Appending this walk with the walk $(a_n, a_{n+1}, \ldots, a_k=a_1)$ yields a directed walk from $a_1$ to $a_1$ of length $k+1$. In this way, we can get a directed walk from $a_1$ to $a_1$ of any length greater than $k$.

    Now we return to the inductive proof and start with the base of induction for $l=0$ or $l=1$.
    If $l=0$, then we have found a loop. If $l=1$ we have a closed walk of length $2$, that is, a pair $(a,b)$ which belongs to both $R$ and $R^{-1}$. We set $R' = R \cap R^{-1}$ and observe that $R'$ is nonempty and symmetric, and it is not hard to verify that $R'$ locally absorbs $=_A$.
    Ol\v s\'ak's loop lemma, in the form of Theorem~\ref{thm:olsak_loop}, gives us a loop in $R$.

    Finally, we make the induction step from $l-1$ to $l$.
    Take a closed walk $(a_1, a_2, \ldots)$ of length $2^l$ and consider
    $R' = \gp{R}{2}$. Observe that $R'$ contains a directed closed walk of length $2^{l-1}$ (namely $(a_1, a_3, \ldots)$), and that $R'$ locally absorbs $=_A$ (by Lemma~\ref{lem:abs_pp}), so, by the inductive hypothesis, $R'$ has a loop. In other words, $R$ has a directed closed walk of length 2 and we are done by the case $l=1$.
\end{proof}

Note that we cannot further relax the assumption on the graph by requiring that, for example, it has an infinite directed walk.
Indeed the natural order of the rationals~(taken for $R$) locally $2$-absorbs the equality relation by the binary arithmetic mean
operation $(a+b)/2$~%
(i.e., all the absorbing evaluations are realized by a single operation).
The same relation locally $4$-absorbs equality with the near unanimity operation  $n(x,y,z,w)$ which, when applied to $a\leq b\leq c\leq d$, in any order, returns $(b+c)/2$.

Nevertheless, we can strengthen the algebraic assumption and still provide a loop; the following theorem is one of the key components in the proof sketch provided in Section~\ref{sect:sensitivity} and the full proof found in Appendix~\ref{sec:k+2}~%
(albeit applied there with $l=1$).

\begin{theorem} \label{thm:loop_inv}
    Let $\alg A$ be an \textbf{idempotent} algebra and $\alg R \leq \alg A^2$ contain a directed walk of length $n-1$.  If $R$ locally $n$-absorbs $=_A$ and $\gp{R}{l}$ locally $n$-absorbs $R^{-1}$ for some $l \in \en$ then $R$ contains a loop.
\end{theorem}

\begin{proof}
By applying Lemma~\ref{lem:aux_loops} similarly as in the proof of Theorem~\ref{thm:loop_cycle},  we can get, from a directed walk of length $n-1$, a directed walk $(a_1, a_2, \ldots)$ of an arbitrary length. Moreover, by the same reasoning, for each $i$ and $j$ with $j \geq i + n - 1$, there is a directed walk from $a_i$ to $a_j$ of any length greater than or equal to $j-i$.

Consider the relations $R' = \gp{R}{ln^2}$ and $S = \gp{(R^{-1})}{n^2}$, and tuples
\begin{align*}
\vc{c} &= (c_1, \ldots, c_n) := (a_{n^2}, a_{(n+1)n}, \ldots a_{(2n-1)n}), \mbox{ and } \\
\vc{d} &= (d_1, \ldots, d_n) := (a_n, a_{2n} \ldots, a_{n^2})
\end{align*}
By the previous paragraph and the definitions, both $\vc{c}$ and $\vc{d}$ are directed walks in $R'$, and $(c_i,d_i) \in S$ for each $i$. Moreover, since $\gp{R}{l}$ locally $n$-absorbs $R^{-1}$, Lemma~\ref{lem:abs_pp} implies that $R'$ locally absorbs $S$.
We can thus apply Lemma~\ref{lem:aux_loops} to the relations $R'$, $S$ and the tuples $\vc{c},\vc{d}$ and obtain a directed walk from $c_1=a_{n^2}$ to $d_{n-1}=a_{n^2}$ in $R'$. This closed walk in turn gives a closed directed walk in $R$ and we are in a position to finish the proof by applying Theorem~\ref{thm:loop_cycle}.
\end{proof} 

\section{Consistent instances are sensitive (sketch of a proof)}\label{sect:sensitivity}

In this section we present the main ideas that are used to prove the ``only if'' direction in Theorem~\ref{thm:mainresult}
and ``1 implies 2'' in Theorem~\ref{thm:sensfull}. These ideas are shown in a very simplified situation, in particular, only the case that $k=2$ and $\alg A$ is finite is considered.
In the end of this section we briefly discuss the necessary adjustments in the general situation.
A complete proof is given in Appendix~\ref{sec:k+2}.

Consider a finite idempotent algebra $\alg A$  with a 4-ary near unanimity term operation and a $(2,3)$-instance $\inst I = (V, \constr)$ over $\alg A$. Each pair $\{x,y\}$ of variables is constrained by a unique constraint $((x,y), R_{xy})$ or $((y,x),R_{yx})$. For convenience we also define $R_{yx} = R^{-1}_{yx}$ (or $R_{xy} = R^{-1}_{yx}$ in the latter case) and $R_{xx}$ to be the equality relation on $A$.
Our aim is to show that every pair in every constraint relation extends to a solution. The overall structure of the proof is by induction on the number of variables of $\inst I$.

We fix a pair of variables $\{x_1,x_2\}$ and a pair $(a_1,a_2) \in R_{x_1x_2}$ that we want to extend.
The strategy is to consider the instance $\inst J$ obtained by removing $x_1$ and $x_2$ from the set of variables and shrinking the constraint relations $R_{uv}$ to $R'_{uv}$ so that only the pairs consistent with the fixed choice remain, that is,
\[
R'_{uv} = \{(b,c) \in R_{uv} \mid (a_1,b) \in R_{x_1u}, (a_2,b) \in R_{x_2u},
(a_1,c) \in R_{x_1v},
(a_2,c) \in R_{x_2v}\}.
\]
We will show that $\inst J$ contains a nonempty $(2,3)$-subinstance, that is, an instance whose constraint relations are nonempty subsets of the original ones. The induction hypothesis then gives us a solution to $\inst J$ which, in turn, yields a solution to $\inst I$ that extends the fixed choice.

Having a nonempty $(2,3)$-subinstance can be characterized by the solvability of certain relaxed instances. The following concepts will be useful for working with relaxations of $\inst I$ and $\inst J$.

\begin{definition}
A \emph{pattern} is a triple $\pat{P} = (W; \asc{F},l)$,
  where $(W; \asc{F})$ is an undirected graph, and $l$ is a mapping $l: W \to V$.
  The variable $l(i)$ is referred to as the \emph{label} of $i$.

A \emph{realization} (\emph{strong realization}, respectively) of $\pat{P}$ is a mapping $\alpha: W \to A$,
  which \emph{satisfies} every edge $\{w_1,w_2\} \in \asc{F}$, that is,
$(\alpha(w_1), \alpha(w_2)) \in R_{l(w_1), l(w_2)}$ ($(\alpha(w_1), \alpha(w_2)) \in R'_{l(w_1), l(w_2)}$, respectively).
  (Strong realization only makes sense if $l(W) \subseteq V \setminus \{x_1,x_2\}$.)

A pattern is (\emph{strongly}) \emph{realizable} if it has a (strong) realization.
\end{definition}

The most important patterns for our purposes are \emph{2-trees},
these are patterns  obtained from the empty pattern by gradually adding triangles (patterns whose underlying graph is the complete graph on 3 vertices)  and merging them along a vertex or an edge to the already constructed pattern. Their significance stems from the following well known fact.

\begin{lemma} \label{lem:trees_and_two_three}
An instance (over a finite domain) contains a nonempty (2,3)-subinstance if and only if every 2-tree is realizable in it.
\end{lemma}

The ``only if'' direction of the lemma applied to the instance $\inst I$ implies that every 2-tree is realizable. The ``if'' direction applied to the instance $\inst J$ tells us that our aim boils down to proving that every 2-tree is strongly realizable.
This is achieved by an induction on a suitable measure of complexity of the tree using several constructions. We will not go into full technical details here, we rather present several lemmata whose proofs contain essentially all the ideas that are necessary for the complete proof.

\begin{lemma} \label{lem:babysr1}
Every edge (i.e., a pattern whose underlying graph is a single edge) is strongly realizable.
\end{lemma}

\begin{figure}
\centering
\begin{minipage}{.45\textwidth}
  \centering
\begin{tikzpicture}
    \node[draw=none] (1) at (-0.5,1) {$W'$};
    \node[shape=circle, draw=black, fill=white, label=left:{$x_1$}] (2) at (1,4) {$w_{11}$};
    \node[shape=circle, draw=black, fill=white, label=left:{$x_2$}] (3) at (1,3) {$w_{12}$};
    \node[shape=circle, draw=black, fill=white, label=left:{$x_1$}] (4) at (1,2) {$w_{21}$};
    \node[shape=circle, draw=black, fill=white, label=left:{$x_2$}] (5) at (1,1) {$w_{22}$};

    \node[fit=(1)(2), draw, dashed, inner sep=5mm] {};

    \node[shape=circle, draw=black, fill=white, label=right:{$z_1$}] (8) at (3,3.5) {$w^1$};
    \node[shape=circle, draw=black, fill=white, label=right:{$z_2$}] (11) at (3,1.5) {$w^2$};

    \draw (8)--(2); \draw (8)--(3); \draw (11)--(4); \draw (11)--(5); \draw (8)--(11);
\end{tikzpicture}

  \captionof{figure}{Pattern $\pat{P}$ in Lemma~\ref{lem:babysr1}.}
  \label{fig:Lemma_baby_sr1}
\end{minipage}%
\begin{minipage}{.55\textwidth}
  \centering
\begin{tikzpicture}
    \node[shape=circle, draw=black, fill=white] (a1) at (1,3) {};
    \node[shape=circle, draw=black, fill=white] (a2) at (3,3) {};
    \node[shape=circle, draw=black, fill=white] (a3) at (5,3) {};
    \node[shape=circle, draw=black, fill=white] (a4) at (7,3) {};

    \node[shape=circle, draw=black, fill=white] (b1) at (2,2) {};
    \node[shape=circle, draw=black, fill=white] (b2) at (4,2) {};
    \node[shape=circle, draw=black, fill=white] (b3) at (6,2) {};

    \node[shape=circle, draw=black, fill=white] (c1) at (1.5,1) {};
    \node[shape=circle, draw=black, fill=white] (c2) at (3.5,1) {};
    \node[shape=circle, draw=black, fill=white] (c3) at (5.5,1) {};

    \node[shape=circle, draw=black, fill=white] (d1) at (2.5,1) {};
    \node[shape=circle, draw=black, fill=white] (d2) at (4.5,1) {};
    \node[shape=circle, draw=black, fill=white] (d3) at (6.5,1) {};

    \draw (a1)--(b1); \draw (a2)--(b2); \draw (a3)--(b3);
    \draw (a2)--(b1); \draw (a3)--(b2); \draw (a4)--(b3);
    \draw (a1)--(c1); \draw (a2)--(c2); \draw (a3)--(c3);
    \draw (a2)--(d1); \draw (a3)--(d2); \draw (a4)--(d3);
    \draw (b1)--(c1); \draw (b2)--(c2); \draw (b3)--(c3);
    \draw (b1)--(d1); \draw (b2)--(d2); \draw (b3)--(d3);
\end{tikzpicture}

  \captionof{figure}{Path of three bow ties.}
  \label{fig:bowties}
\end{minipage}
\end{figure}

\begin{proof}[Proof sketch]
Let $\pat{Q}$ be the pattern formed by an undirected edge with vertices $w^1$ and $w^2$ labeled $z_1$ and $z_2$, respectively.
Let $\pat{P}$ be the pattern obtained from $\pat{Q}$ by adding a set of four fresh vertices
  $W' = \{w_{11}, w_{12}, w_{21}, w_{22}\}$ labeled $x_1, x_2, x_1, x_2$, respectively,
  and adding the edges $\{w^i,w_{i1}\}$ and  $\{w^i,w_{i2}\}$ for $i=1,2$, see Figure~\ref{fig:Lemma_baby_sr1}.
Observe that the restriction of a realization $\beta$ of $\pat{P}$,
  such that $\beta(w_{ij}) = a_j$ for each $i,j \in \{1,2\}$, to the set $\{w^1,w^2\}$ is a strong realization of $\pat{Q}$.

We consider the set $T$ of restrictions of realizations of $\pat{P}$ to the set $W'$.
  Since constraint relations are subuniverses of $\alg A^2$, it follows that $T$ is a subuniverse of $\alg A^4$.
\[
T = \{ (\beta(w_{11}),\beta(w_{12}),\beta(w_{21}),\beta(w_{22})) \mid
     \beta \mbox{ realizes $\pat{P}$}\} \leq \alg A^4
\]
We need to prove that the tuple $\vc{a}=(a_1,a_2,a_1,a_2)$ is in $T$.
By the Baker-Pixley theorem, Theorem~\ref{thm:BPold}, it is enough to show that for any $3$-element set of coordinates, the relation $T$ contains a tuple  that agrees with $\vc{a}$ on this set. This is now our aim.

For simplicity, consider the set of the first three coordinates.
We will build a realization $\beta$ of $\pat P$ in three steps. After each step, $\beta$ will satisfy all the edges where it is defined.
First, since $(a_1,a_2) \in R_{x_1x_2}$ and $\inst I$ is a (2,3)-instance, we can find $b_1 \in A$ such that $(a_1,b_1) \in R_{x_1z_1}$ and $(a_2,b_1) \in R_{x_2z_1}$, and we
set $\beta(w_{11}) = a_1$, $\beta(w_{12}) = a_2$, and $\beta(w^1) = b_1$.
Second, we find $b_2 \in A$ such that $(a_1,b_2) \in R_{x_1z_2}$ and $(b_1,b_2) \in R_{z_1z_2}$ (here we use $(a_1,b_1) \in R_{x_1z_1}$ and that $\inst I$ is a (2,3)-instance), and
set $\beta(w_{21}) = a_1$, $\beta(w^2) = b_2$.
Third, using $(a_1,b_2) \in R_{x_1z_2}$ we find $a_2'$ such that $(b_2,a_2') \in R_{z_2x_2}$ and set $\beta(w_{22})= a_2'$.
By construction, $\beta$ is a realization of $\pat{P}$ and
$(\beta(w_{11}),\beta(w_{12}),\beta(w_{21})) = (a_1,a_2,a_1)$, so our aim has been achieved.
\end{proof}

\noindent Using Lemma~\ref{lem:babysr1}, one can go a step further and prove that every pattern built on a graph which is a triangle is strongly realizable.
We are not going to prove this fact here.


\begin{lemma} \label{lem:babysr2}
Every bow tie (a pattern whose underlying graph is formed by two triangles with a single common vertex) is strongly realizable.
\end{lemma}

\begin{proof}[Proof sketch]
  Let $\rel W'_1$ and $\rel W'_2$ be two triangles~(viewed as undirected graphs)
  with a single common vertex $w$.
  Let $\pat{Q'}$ be any pattern over $W'_1\cup W'_2$ with labelling $l'$ sending $W'_1\cup W'_2$ to $V\setminus\{x_1,x_2\}$.
  Similarly as in the proof of Lemma~\ref{lem:babysr1} we form a pattern $\rel Q$ by adding to $\rel Q'$ ten additional vertices
  (five of them labeled $x_1$, the other five $x_2$)
  and edges so that the restriction of a realization $\alpha$ of $\pat Q$ to the set $W_1'\cup W_2'$ is a strong realization of $\pat Q'$
  whenever the additional vertices have proper values (that is, value $a_i$ for vertices labeled $x_i$).

  We will gradually construct a realization $\alpha$ of $\pat{Q}$,
  which sends all the vertices labeled by $x_1$ to $a_1$,
  and all the vertices labeled by $x_2$ and adjacent to a vertex in $W'_1$ to $a_2$.
  First use the discussion after Lemma~\ref{lem:babysr1} to
  find a strong realization of $\pat Q'$ restricted to $W_1'$.
  This defines $\alpha$ on $W_1'$ and its adjacent vertices labeled by $x_1$ and $x_2$.

  Next, we want to use Lemma~\ref{lem:babysr1} for assigning values to the two remaining vertices of $W'_2$.
  However, in order to accomplish that, we need to shift the perspective:
  the role of $x_1$ is played by $x_1$,
  but the role of $x_2$ is played by $l'(w)$;
  and the role of $(a_1,a_2)$ is played by $(a_1,\alpha(w))$.
  In this new context, we use Lemma~\ref{lem:babysr1} to find a strong realization of  the edge-pattern formed by the two remaining vertices of $W_2'$~
  (with a proper restriction of $l'$).
  This defines $\alpha$ on all the vertices of $\pat Q$, except for the two vertices adjacent to $W_2'\setminus\{w\}$ and labeled by $x_2$.
  Finally, similarly as in the third step in the proof of Lemma~\ref{lem:babysr1}, we define $\alpha$ on the remaining two vertices (labeled $x_2$)
  to get a sought after realization of $\rel Q$.

  Now $\alpha$ assigns proper values ($a_1$ or $a_2$) to all additional vertices,
  except those two coming from the non-central vertices of $W'_2$ and labeled by $x_2$.
  We apply the 4-ary  near unanimity term operation
  to the realization $\alpha$ and its 3 variants obtained by exchanging the roles of $W'_1$ and $W_2'$ and $x_1$ and $x_2$.
  The result of this application is a realization of $\pat Q$ which defines a strong realization of $\pat Q'$.
\end{proof}

\noindent In the same way it is possible to prove strong realizability of further patterns, such as those in the following corollary.

\begin{corollary} \label{cor:babypppp}
Every ``path of 3 bow ties'' (i.e., a pattern whose underlying graph is as in Figure~\ref{fig:bowties}) is strongly realizable.
\end{corollary}

The application of the loop lemma is illustrated by the final lemma in this section.

\begin{lemma}
Every diamond (i.e., a pattern whose underlying graph is formed by two triangles with a single common edge) is strongly realizable.
\end{lemma}

\begin{proof}[Proof sketch]
  The idea is to merge two vertices in a bow tie using the loop lemma.
  Let $\pat{Q'}$ be a pattern over a graph which is a bow tie on two triangles $W_1'$ and $W_2'$~%
  (just like in the proof of Lemma~\ref{lem:babysr2}).
  Let $w_1\in W_1'\setminus W_2'$ and $w_2\in W_2'\setminus W_1'$ be such that $l(w_1)=l(w_2)$.

  Let $\pat{Q}$ be obtained from $\pat Q'$ exactly as in the proof of Lemma~\ref{lem:babysr2}
  and notice that a proper realization $\alpha$ of $\pat{Q}$ with $\alpha(w_1)=\alpha(w_2)$ gives us a strong realization of a diamond.
    Let  $\pat{Q}^{3}$ be the pattern obtained by taking
the disjoint union of $3$ copies of $\pat{Q}$ and identifying the vertex $w_2$ in the $i$-th copy with the vertex $w_1$ in the $(i+1)$-first copy,
  for each $i \in \{1,2\}$ (Figure~\ref{fig:bowties} shows $\rel Q^3$ without the additional vertices).

Denote by $T$ the set of all the  realizations $\beta$  of  $\pat{Q}$
and denote by $S \subseteq T$  the set of those $\beta \in T$ that are proper.
By a straightforward argument, both $T$ and $S$ are subuniverses of $\prod_{w \in Q} \alg A$.
Using the near unanimity term operation of arity $4$,
 $S$ clearly $4$-absorbs $T$.

The plan is to apply Theorem~\ref{thm:loop_inv} to the binary relation $\proj_{w_1,w_2} S \subseteq A \times A$.
  As noted above, a loop in this relation gives us the desired strong realization of a diamond,
  so it only remains to verify the assumptions of Theorem~\ref{thm:loop_inv}.
By Corollary~\ref{cor:babypppp}, the pattern $\pat{Q}^3$ has a proper realization.
  The images of copies of vertices $w_1$ and $w_2$ in such a realization yield a directed walk in $\proj_{w_1,w_2}(S)$ of length $3$.
 Next, since $S$ $4$-absorbs $T$, then $\proj_{w_1,w_2}(S)$  $4$-absorbs $\proj_{w_1,w_2}(T)$,
  so it is enough to verify that the latter relation contains $=_A$ and $\proj_{w_1,w_2}(S)^{-1}$.
 We only look at the latter property.
  Consider any $(b_1,b_2) \in \proj_{w_1,w_2}(S)^{-1}$.
  By the definition of $S$, the pattern $\pat{Q}$ has a realization $\alpha$ such that $\alpha(w_1)=b_2$ and $\alpha(w_2)=b_1$.
  We flip the values $\alpha(w_1)$ and $\alpha(w_2)$,
  restrict $\alpha$ to $\{w_1,w_2\}$ together with the middle vertex of the bow tie,
  and then extend this assignment  to a  realization of $\pat{Q}$, giving us $(b_1,b_2) \in \proj_{w_1,w_2}(T)$.
\end{proof}

There are two major adjustments needed for the general case.
First, the ``if'' direction of Lemma~\ref{lem:trees_and_two_three} (and its analogue for a general $k$) is no longer true over infinite domains. This is resolved by working directly with the realizability of $k$-trees and proving a more general claim by induction: instead of ``a $(k,k+1)$-instance is sensitive'' we prove, roughly, that any evaluation, which extends to a sufficiently deep $k$-tree, extends to a solution.
Second, for higher values of $k$ than 2 we do not prove strong realizability in one step as in, e.g., Lemma~\ref{lem:babysr1}, but rather go through a sequence of intermediate steps between realizability and strong realizability.

\section{Conclusion}
\label{sec:conclusions}

\newtheorem{conjecture}[theorem]{Conjecture}
\newtheorem{problem}[theorem]{Problem}

We have characterized varieties that have sensitive $(k,k+1)$-instances of the $\csp$ as those that possess a near unanimity term of arity $k+2$.
From the computational perspective, the following corollary is perhaps the most interesting consequence of our results.

\begin{corollary} \label{cor:one}
Let $\rel{A}$ be a finite $\csp$ template whose  relations all have arity at most $k$ and which has a near unanimity polymorphism of arity $k+2$. Then every instance of the $\csp$ over $\rel{A}$, after enforcing  $(k,k+1)$-consistency, is sensitive.
\end{corollary}

Therefore not only is the $(k,k+1)$-consistency algorithm  sufficient to detect global inconsistency, we also additionally get the sensitivity property. 
Let us compare this result to some previous results as follows. Consider a template $\rel{A}$ that, for simplicity, has only unary and binary relations and that has a near unanimity polymorphism of arity $k+2 \geq 4$. Then any instance of the $\csp$ over $\rel{A}$ satisfies the following. 

\begin{enumerate}
    \item After enforcing $(2,3)$-consistency, if no contradiction is detected, then the instance has a solution~\cite{cdbw} (this is the bounded width property).
    \item After enforcing $(k,k+1)$-consistency, every partial solution on $k$ variables extends to a solution (this is the sensitivity property). 
    \item After enforcing $(k+1,k+2)$-consistency, every partial solution extends to a solution~\cite{feder-vardi} (this is the bounded strict width property).
\end{enumerate}

For $k+2 > 4$ there is a gap between the first and the second item. Are there natural conditions that can be placed there?

The properties of a template $\rel{A}$ from the first and the third item (holding for every instance) can be characterized by the existence of certain polymorphisms: a near unanimity polymorphism of arity $k+2$ for the third item~\cite{feder-vardi} and weak near unanimity polymorphisms of all arities greater than 2 for the first item~\cite{barto-kozik-bw,bulatov-boundedwidth,kkvw}. This paper does not give such a direct characterization for the second item (essentially, since Theorem~\ref{thm:sensfull} involves a square). Is there any? Moreover, there are characterizations for natural extensions of the first and the third to relational structures with higher arity relations~\cite{feder-vardi,barto-hierarchy}. This remains open for the second item as well.

In parallel with the flurry of activity around the $\csp$ over finite templates, there has been much work done on the $\csp$ over infinite $\omega$-categorical templates~\cite{bodirsky-phd,Pin15}. These templates cover a much larger class of computational problems but, on the other hand, share some pleasant properties with the finite ones. In particular, the $(k,k+1)$-consistency of an instance can still be enforced in polynomial time. Corollary~\ref{cor:one} can be extended to this setting as follows.

\begin{corollary}
Let $\rel{A}$ be an $\omega$-categorical $\csp$ template whose  relations all have arity at most $k$ and which has local idempotent near unanimity polymorphisms of arity $k+2$. Then every instance of the $\csp$ over $\rel{A}$, after enforcing the $(k,k+1)$-consistency, is sensitive.
\end{corollary}

Bounded strict width $k$ of an $\omega$-categorical template was characterized in \cite{infinite-strict-width} by the existence of  
a \emph{quasi-near unanimity} polymorphism $n$ of arity $k+1$, i.e.,  
\[
n(y,x,\dots, x) \approx n(x,y,\dots, x) \approx \dots \approx n(x,x, \dots, y) \approx n(x,x, \dots, x),
\]
which is, additionally, \emph{oligopotent}, i.e., the unary operation $x \mapsto n(x,x, \dots, x)$ is equal to an automorphism on every finite set. This result extends the characterization of Feder and Vardi since an oligopotent quasi-near unanimity polymorphism generates a near unanimity polymorphism as soon as the domain is finite. On an infinite domain, however, oligopotent quasi-near unanimity polymorphisms generate local near unanimity polymorphisms which, unfortunately, do not need to be idempotent on the whole domain. Our results thus fall short of proving the following natural generalization of Corollary~\ref{cor:one} to the infinite.

\begin{conjecture}
  Let $\rel A$ be an $\omega$-categorical $\csp$ template
  whose  relations all have arity at most $k$ and 
  which has an oligopotent quasi-near unanimity polymorphism of arity $k+2$. 
  Then every instance of the $\csp$ over $\rel{A}$, after enforcing $(k,k+1)$-consistency, is sensitive.
\end{conjecture}

To confirm the conjecture, a new approach, that does not use a loop lemma, will be needed since there are examples of $\omega$-categorical structures having oligopotent quasi-near unanimity polymorphisms for which the counterpart to Theorem~\ref{thm:olsak_loop} does not hold. Indeed, one such an example is the infinite clique.



\bibliography{biblio}


\appendix

\section{Proofs of Theorems \ref{thm:BP} and \ref{thm:localBP}}\label{sec:AlgebraProofs}

\noindent The first result is due to Bergman~\cite{Bergman1977}, we provide a short proof for the convenience of the reader.
\begin{customthm}{\ref{thm:BP}}
  Let $k > 1$ and   $\var V$ be a variety.  The following are equivalent:
  \begin{enumerate}
  \item $\var V$ has a $(k+1)$-ary near unanimity term;
  \item  any partial solution of a $(k,k+1)$-instance over $\var V$  extends to a solution.
  \end{enumerate}
\end{customthm}
\begin{proof}[Proof of Theorem~\ref{thm:BP}]
A straightforward modification of the ``if'' direction of the proof of Theorem~\ref{thm:mainresult}, using Proposition~\ref{prop:sw} in place of Proposition~\ref{prop:sens} shows that the second condition implies the existence of a $(k+1)$-ary near unanimity term (also see~\cite[Lemma 11]{Bergman1977}).  For the converse,
suppose that $\var V$ has a $(k+1)$-ary near unanimity term $n(x_1, \dots, x_{k+1})$ and let $\inst I = (V, \constr)$ be a $(k, k+1)$-instance of $\csp(\var V)$.

Let $n =|V|$.  We will show by induction on $r < n$ that if $W \subseteq V$ with $|W| = r$ then any solution of $\inst I_{|W}$ can be extended to a solution of $\inst I_{|W \cup \{v\}}$ for any $v \in V\setminus W$.  From this, the implication will follow. By the assumption that $\inst I$ is a $(k, k+1)$-instance it follows that this property holds for $r = k$.  So, assume that $k < r < n$ and suppose that $W \subseteq V$ with $|W| = r$.  Let $v \in V \setminus W$ and let $f$ be a solution of $\inst I_{|W}$.

Fix some listing of the elements of $W$, say $W = \{v_1, v_2, \dots, v_r\}$ and for $1 \le i \le r$ let  $W_i = (W\setminus \{v_i\})\cup \{v\}$.  By induction, there is a solution $f_i$ of $\inst I_{|W_i}$ that extends the restriction of $f$ to $W\setminus \{v_i\}$, for $1 \le i \le k+1$.   We claim that the extension of $f$ to $W\cup\{v\}$ by setting $f(v) = n(f_1(v), f_2(v), \dots, f_{k+1}(v))$ produces a solution of $\inst I_{|W \cup\{v\}}$.

We need to show that if $U \subseteq W \cup \{v\}$ with $|U| = k$ then $(f(u)\mid u \in U)$ satisfies the unique constraint $(U,R)$ of $\inst I$ with scope $U$.  When  $U\subseteq W$, this is immediate, so assume that $v\in U$.  For $1 \le i \le k+1$, let $g_i$ be the restriction of $f_i$ to $U$, if $v_i \notin U$ and otherwise let $g_i$ be some partial solution of $\inst I_{|U}$ that extends the restriction of $f_i$ to $U \setminus \{v_i\}$.  Since each $g_i$ satisfies the constraint  $(U,R)$ then so does $n(g_1, g_2, \dots, g_{k+1})$. Using that $n$ is a near unanimity term it can be shown that this element is equal to $f_{|U}$, as required.
\end{proof}

The next theorem is a variation of the Baker-Pixley~\cite{BakerPixley} result for idempotent, not necessarily finite, algebras.

\begin{customthm}{\ref{thm:localBP}}
Let $\alg A$ be an idempotent algebra and $k > 1$.  The following are equivalent:
\begin{enumerate}
\item $\alg A$ has local near unanimity term operations of arity $k+1$;
\item for every $r > k$, every  subalgebra of $\alg A^r$ is uniquely determined by its projections onto all $k$-element subsets of coordinates;
\item   every subalgebra of $\alg A^{k+1}$ is uniquely determined by its projections onto all $k$-element subsets of coordinates.
\end{enumerate}
\end{customthm}

\begin{proof}[Proof of Theorem~\ref{thm:localBP}]
To show that Condition 1 implies Condition 2, suppose that $\alg A$ has local near unanimity term operations of arity $k+1$ and let $\alg R$ be a subalgebra of $\alg A^r$ for some $r >k$.
 Let $\vc{a} = (a_1, \dots, a_{r}) \in  A^{r}$ be a tuple such that for every subset $I$ of $[r]$ of size $k$, there is some element $\vc{b} \in R$ with $\proj_I(\vc{a}) = \proj_I(\vc{b})$.  We will show by induction on $n\ge k$ that if $n \le r$ then for every subset $J$ of $[r]$ of size $n$ there is some $\vc{b} \in R$ with $\proj_J(\vc{a}) = \proj_J(\vc{b})$.  With $n = r$ we conclude that $\vc{a} \in R$, as required.

 By assumption, this property holds when $n = k$.  Suppose that it has been established for some $n$ with $k \le n < r$ and let $J$ be a subset of $[r]$ of size $n+1$.  By symmetry it suffices to consider the case when $J = \{1, 2, \dots, n+1\}$.  For each $i$, with $1 \le i \le k+1$, let $\vc{b}_i \in R$ be such that $\vc{a}$ and $\vc{b}_i$ agree on the set $J \setminus \{i\}$.
 Let $n(x_1, \dots, x_{k+1})$ be a $(k+1)$-ary local near unanimity term operation of $\alg A$ for the subset of $A$ consisting of all of the components of the tuples $\vc{b}_i$, for $1 \le i \le k+1$.  A straightforward calculation shows that $\vc{b} = n(\vc{b}_1, \dots, \vc{b}_{k+1}) \in R$ has the desired property.

Clearly Condtion 2 implies Condition 3. For the remaining implication, we use Corollary 2.7 from \cite{horowitz-ijac} that shows that if $\alg A$ is finite (and idempotent) then it will have a $(k+1)$-ary near unanimity term operation if and only if for every $a_i$, $b_i \in A$, for $1 \le i \le k+1$, there is some term operation $t$ of $\alg A$ such that
\begin{align*}
t(b_1,a_1,a_1, \dots, a_1) &= a_1\\
t(a_2,b_2,a_2, \dots, a_2) &= a_2\\
&\vdots\\
t(a_{k+1},a_{k+1},a_{k+1}, \dots, b_{k+1}) &= a_{k+1}.
\end{align*}
It can be seen from the proof of this result that  if $\alg A$ is not assumed to be finite, then one can conclude that it has local near unanimity term operations of arity $k+1$ if and only if this condition holds for all $a_i$ and $b_i$.

This local term condition can be translated into a statement about subalgebras of $\alg A^{k+1}$, namely that for every $a_i$, $b_i \in A$,
for $1 \le i \le k+1$, the  $(k+1)$-tuple $\vc{a} = (a_1, \dots, a_{k+1})$ belongs to the subalgebra $\alg R$ of $\alg A^{k+1}$ generated by the set of $k+1$ tuples
\[
\{(b_1,a_2,a_3, \dots, a_{k+1}), (a_1,b_2,a_3, \dots, a_{k+1}), \dots, (a_1,a_2,a_3, \dots, b_{k+1})\}.
\]
Our assumption on $\alg A$ guarantees that $\vc{a}$ belongs to $R$ since any projection of $R$ onto $k$ coordinates will contain the corresponding projection of $\vc{a}$.  Thus $\alg A$ will have local near unanimity term operations of arity $k+1$.
\end{proof}

\section{Proof of Theorem~\ref{thm:olsak_loop}}
\label{app:ll}
In this section we present a proof of Theorem~\ref{thm:olsak_loop}.
The proof is a trivial adaptation of reasoning attributed to Ralph McKenzie in~\cite{O17}.
\begin{customthm}{\ref{thm:olsak_loop}}
    Let $\alg A$ be an idempotent algebra and 
    $\alg R\leq \alg A^2$ be nonempty and symmetric.
    If $R$ locally absorbs $=_A$, then $R$ contains a loop.
\end{customthm}
The remaining part of this section is devoted to a proof of Theorem~\ref{thm:olsak_loop} by the way of contradiction.

Let $n$ denote the arity of the absorbing operations.
We choose a counterexample to the theorem minimal with respect to $n$.
Then, we fix an algebra $\alg A$ and will call an $\alg R\leq \alg A^2$
a {\em counterexample candidate} if it is non-empty, symmetric, locally $n$-absorbs $=_A$ and has no loop.
\begin{claim}
    Every counterexample candidate has a closed walk of odd length.
\end{claim}
\begin{proof}
    Since $R$ is nonempty and symmetric we have $(a,b),(b,a)\in R$. 
    Apply Lemma~\ref{lem:aux_loops} to the walk $(a,b,a,b,\dotsc,a/b)$ 
    of length $n-1$~(i.e., $n$ vertices, $n-1$ steps) taken twice~(where the last element is either $a$ or $b$ depending on the parity of $n$).
    The lemma provides a directed walk of length $n$
    connecting the first and last elements.
    Since $R$ is symmetric all the edges are undirected and we obtained a closed walk of odd length.
\end{proof}
\begin{claim}
    There exists a counterexample candidate containing a $3$-element clique.
\end{claim}
\begin{proof}
    Take a counterexample candidate $R$; 
    it has an odd cycle, and if it has a triangle we are done. 
    Thus the length of a shortest odd cycle is greater than $3$.
    In this case, however $R\circ R\circ R$ is a counterexample candidate~%
    (we use Lemma~\ref{lem:abs_pp} to provide local absorption) with shorter odd cycle.
    We proceed this way and, in the end, find a counterexample candidate with a $3$-cycle~(which is a $3$-clique).
\end{proof}

\begin{claim}
    No counterexample candidate contains an $n$-element clique.
\end{claim}
\begin{proof}
    Suppose $(a_1,\dotsc,a_n)$ is such a clique.
    We can choose, using the definition of local absorption, $t$ such that $(t(a_1,\dotsc,a_n),t(a_i,\dotsc,a_i))\in R$ for all $i$.
    We use this fact, and the fact that $\alg R\leq \alg A^2,$ to conclude that
    \begin{equation*}
    \big(t(t(a_1,\dotsc,a_n),\dotsc,t(a_1,\dotsc,a_n)),t(a_1,\dotsc,a_n)\big)\in R,
    \end{equation*}
    but, by the idempotency of $t$, the two elements are equal and we have obtained a loop --- a contradiction.
\end{proof}

In order to finish the proof we fix $R$ to be a counterexample candidate with a $3$-element clique and let $a_1,\dotsc,a_m$ be distinct, forming a maximal clique in $R$~%
(such a clique exists by the last claim).
Let $B$ be the subset of $A$ containing vertices with edges to each of $a_1,\dotsc,a_{m-2}$.
Note that $B$ is a subuniverse~%
(since $\alg A$ is idempotent)
and  $S = B^2\cap R$ is nonempty
as $(a_{m-1},a_m),(a_m,a_{m-1})\in S$.

Note, that $\alg S\leq \alg B^2$ is symmetric, nonempty, has no $3$-clique and it locally $n$-absorbs $=_B$.
We obtain a contradiction by showing that $\alg T = \alg S\circ\alg S\circ\alg S$ locally $n-1$ absorbs $=_B$.
The graph $T$ is non-empty, symmetric, has no loop and $\alg T\leq\alg B^2$. 
We will fix a one-$=_B$-in-$T$ tuple, and construct a finite set of one-$=_A$-in-$R$ tuples such that if $t(x_1,\dotsc,x_n)$ is an operation of $\alg A$ producing elements of $R$ on the tuples from the last set then $t(x_1,x_1,x_2,x_3,\dotsc,x_{n-1})$ produces an element of $T$ on the original tuple.
The theorem we are working to prove clearly follows from this fact.

Let $(c_1,d_1),(c_2,d_2),\dotsc,(c_{n-1},d_{n-1})$ be a one-$=_B$-in-$T$ tuple. We consider two cases: in case one $c_i=d_i$ for some $i>1$ and in case two $c_1=d_1$.
In case one (see Figure~\ref{fig:case1}), we assume, wlog that $i=2$, and find for all $j\neq 2$ elements $c'_j,c''_j$ such that 
$(c_j,c'_j),(c'_j,c''_j),(c''_j,d_j)\in S$.
It suffices to take care of the three following one-$=_A$-in-$R$ evaluations:

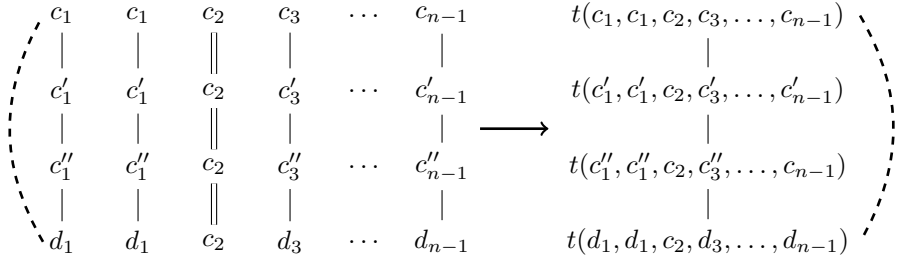
\begin{figure}
    \centering
    \begin{tikzpicture}
        
        \node (00) at (0,0) {$c_1$};
        \node (10) at (1,0) {$c_1$};
        \node (20) at (2,0) {$c_2$};
        \node (30) at (3,0) {$c_3$};
        \node (40) at (4,0) {$\dots$};
        \node (50) at (5,0) {$c_{n-1}$};
        \node (80) at (8.5,0) {$t(c_1,c_1,c_2,c_3,\dots ,c_{n-1})$};
        
        \node (01) at (0,-1) {$c'_1$};
        \node (11) at (1,-1) {$c'_1$};
        \node (21) at (2,-1) {$c_2$};
        \node (31) at (3,-1) {$c'_3$};
        \node (41) at (4,-1) {$\dots$};
        \node (51) at (5,-1) {$c'_{n-1}$};
        \node (81) at (8.5,-1) {$t(c'_1,c'_1,c_2,c'_3,\dots ,c'_{n-1})$};
        
        \node (02) at (0,-2) {$c''_1$};
        \node (12) at (1,-2) {$c''_1$};
        \node (22) at (2,-2) {$c_2$};
        \node (32) at (3,-2) {$c''_3$};
        \node (42) at (4,-2) {$\dots$};
        \node (52) at (5,-2) {$c''_{n-1}$};
        \node (82) at (8.5,-2) {$t(c''_1,c''_1,c_2,c''_3,\dots ,c_{n-1})$};
        
        \node (03) at (0,-3) {$d_1$};
        \node (13) at (1,-3) {$d_1$};
        \node (23) at (2,-3) {$c_2$};
        \node (33) at (3,-3) {$d_3$};
        \node (43) at (4,-3) {$\dots$};
        \node (53) at (5,-3) {$d_{n-1}$};
        \node (83) at (8.5,-3) {$t(d_1,d_1,c_2,d_3,\dots ,d_{n-1})$};

        \path [-,line width=1pt, dashed] (-0.25,-3) edge[bend left] node [left] {} (-0.25,0);
        
        \path [-] (00) edge node[left] {} (01);
        \path [-] (01) edge node[left] {} (02);
        \path [-] (02) edge node[left] {} (03);
        
        \path [-] (10) edge node[left] {} (11);
        \path [-] (11) edge node[left] {} (12);
        \path [-] (12) edge node[left] {} (13);
        
        \draw[double equal sign distance] (20) -- (21);
        \draw[double equal sign distance] (21) -- (22);
        \draw[double equal sign distance] (22) -- (23);
        
        \path [-] (30) edge node[left] {} (31);
        \path [-] (31) edge node[left] {} (32);
        \path [-] (32) edge node[left] {} (33);
        
        \path [-] (50) edge node[left] {} (51);
        \path [-] (51) edge node[left] {} (52);
        \path [-] (52) edge node[left] {} (53);
        
        \path [-] (80) edge node[left] {} (81);
        \path [-] (81) edge node[left] {} (82);
        \path [-] (82) edge node[left] {} (83);
        
        \path [-,line width=1pt, dashed] (10.5,0) edge[bend left] node [left] {} (10.5,-3);
        
        \path [->, line width=1pt] (5.5,-1.5) edge node[left] {} (6.4,-1.5);
        
    \end{tikzpicture}
    \caption{Solid lines are are $S$-related and dashed lines are $T$-related.}
    \label{fig:case1}
\end{figure}

\begin{gather*}
(c_1,c'_1),(c_1,c'_1),(c_2,c_2),(c_3,c'_3),\dotsc,(c_{n-1},c'_{n-1}),
\\
(c'_1,c''_1),(c'_1,c''_1),(c_2,c_2),(c'_3,c''_3),\dotsc,(c'_{n-1},c''_{n-1})\text{ and }\\
(c''_1,d_1),(c''_1,d_1),(c_2,c_2),(c''_3,d_3),\dotsc,(c''_{n-1},d_{n-1}).
\end{gather*}
In case two (see Figure~\ref{fig:case2}) the situation is a bit more involved, we define $c'_i,c''_i$ for all $i>1$ but need $4$ evaluations:

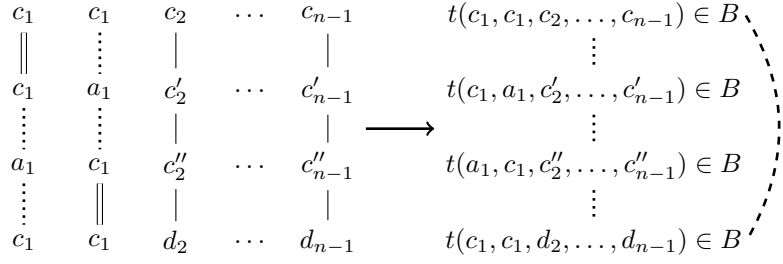
\begin{figure}
    \centering
    \begin{tikzpicture}
        
        \node (00) at (0,0) {$c_1$};
        \node (10) at (1,0) {$c_1$};
        \node (20) at (2,0) {$c_2$};
        \node (30) at (3,0) {$\dots$};
        \node (40) at (4,0) {$c_{n-1}$};
        \node (70) at (7.5,0) {$t(c_1,c_1,c_2,\dots ,c_{n-1})\in B$};
        
        \node (01) at (0,-1) {$c_1$};
        \node (11) at (1,-1) {$a_1$};
        \node (21) at (2,-1) {$c'_2$};
        \node (31) at (3,-1) {$\dots$};
        \node (41) at (4,-1) {$c'_{n-1}$};
        \node (71) at (7.5,-1) {$t(c_1,a_1,c'_2,\dots ,c'_{n-1})\in B$};
        
        \node (02) at (0,-2) {$a_1$};
        \node (12) at (1,-2) {$c_1$};
        \node (22) at (2,-2) {$c''_2$};
        \node (32) at (3,-2) {$\dots$};
        \node (42) at (4,-2) {$c''_{n-1}$};
        \node (72) at (7.5,-2) {$t(a_1,c_1,c''_2,\dots ,c''_{n-1})\in B$};
        
        \node (03) at (0,-3) {$c_1$};
        \node (13) at (1,-3) {$c_1$};
        \node (23) at (2,-3) {$d_2$};
        \node (33) at (3,-3) {$\dots$};
        \node (43) at (4,-3) {$d_{n-1}$};
        \node (73) at (7.5,-3) {$t(c_1,c_1,d_2,\dots ,d_{n-1})\in B$};
        
        \draw[double equal sign distance] (00) -- (01);
        \draw[dotted, line width=1pt] (01) -- (02);
        \draw[dotted, line width=1pt] (02) -- (03);
        
        \draw[dotted, line width=1pt] (10) -- (11);
        \draw[dotted, line width=1pt] (11) -- (12);
        \draw[double equal sign distance] (12) -- (13);
        
        \draw[-] (20) -- (21);
        \draw[-] (21) -- (22);
        \draw[-] (22) -- (23);
        
        \draw[-] (40) -- (41);
        \draw[-] (41) -- (42);
        \draw[-] (42) -- (43);
        
        \draw[dotted, line width=1pt] (70) -- (71);
        \draw[dotted, line width=1pt] (71) -- (72);
        \draw[dotted, line width=1pt] (72) -- (73);
        
        \path [-,line width=1pt, dashed] (9.5,0) edge[bend left] node [left] {} (9.5,-3);
        
         \path [->, line width=1pt] (4.5,-1.5) edge node[left] {} (5.4,-1.5);
        
    \end{tikzpicture}
    \caption{Solid lines are are $S$-related, dashed lines are $T$-related, and dotted lines are $R$-related.}
    \label{fig:case2}
\end{figure}

\begin{gather*}
(c_1,c_1),(c_1,a_1),(c_2,c'_2)\dotsc,(c_{n-1},c'_{n-1}),
\\
(c_1,a_1),(a_1,c_1),(c'_2,c''_2),\dotsc,(c'_{n-1},c''_{n-1}),\\
(a_1,c_1),(c_1,c_1),(c''_2,d_2),\dotsc,(c''_{n-1},d_{n-1})\text{ and two new ones }\\
(c_1,a_1),(a_1,a_1),(c'_2,a_1),\dotsc, (c'_{n-1},a_1),\\
(a_1,a_1),(c_1,a_1),(c''_2,a_1),\dots, (c''_{n-1},a_1).
\end{gather*}
The list contains $5$ evaluations, but the second one ~(included for simplicity) is in fact not a one-$=_A$-in-$R$ evaluation, but a usual application of the term to elements of $R$.
Any term, putting all these evaluations in $R$ puts~(by idempotency and the fact that all considered elements are adjacent to $a_i$ if $1<i<m-1$)
$t(c_1,a_1,c'_2,\dotsc,c'_{n-1}), t(a_1,c_1,c''_2,\dotsc,c''_{n-1})\in B$.
These elements witness the path required to put the pair $(t(c_1,c_1,c_2,c_3,\dotsc,c_{n-1}),t(c_1,c_1,d_2,d_3\dotsc,d_{n-1}))$ in $T$.

\section{Consistent instances are sensitive}\label{sec:k+2}
In this section we provide a proof for the ``only if'' direction in Theorem~\ref{thm:mainresult} 
and ``1 implies 2'' in Theorem~\ref{thm:sensfull}.
We will proceed with the two proofs in parallel; 
in one case we fix an algebra $\alg A$ and in the other a variety $\var V$.
We will assume, without loss of generality, that the only  operation symbol of $\var V$ is  $(k+2)$-ary and is a near unanimity operation for all members of $\var V$.  So, all members of $\var V$ are idempotent.
Formally, in the case of Theorem~\ref{thm:sensfull}, we should be working with instances over $\alg A^2$,
but if $\alg A$ has local $(k+2)$-ary near unanimity term operations, then so does $\alg A^2$ and
so we can work directly with an algebra possessing  local near unanimity term operations and denote it by $\alg A$.
We will remark on the differences between these two cases only in the places where we apply near unanimity operations.

For the purpose of this section we modify the definition of an instance slightly:
an instance is a triple $\inst I = (V, \{\alg A_x \mid x \in V\}, \constr) $,  where $\constr = \{(S,R_S) \mid S \subseteq V, |S| \leq k\}$ and $R_S \leq \prod_{x \in S} \alg A_x$.
Note that the definition of a $\csp$ instance is, formally, different than our standard definition:
the variables involved in a constraint are a set and not a tuple.
This minor modification will allow us to present the proofs more succinctly. 
In order for the interpretation of a constraint to be unique we assume, without loss of generality, 
that the algebras $\alg A_x$ are disjoint.
When applying the results of this section in Theorem~\ref{thm:mainresult} we will set each $\alg A_x$ to be an isomorphic copy of $\alg A$,
and in case of Theorem~\ref{thm:sensfull} we will choose isomorphic copies from the variety, so that their domains are disjoint.

The rough idea of the proof is to fix, in a $(k,k+1)$-instance, 
a tuple from the relation constraining set of variables $Y$
and consider the instance obtained by removing $Y$ from the set of variables 
and shrinking the constraint relations so that only the tuples extending the fixed choice of values for the variables in $Y$ remain. 
If we were able to show that the obtained instance contains a $(k,k+1)$-subinstance, 
both theorems would then follow by induction on the number of variables of the instance. 
It is well known that for instances \textbf{with finite domains}, 
the latter property is equivalent to the solvability of certain relaxed instances, here called $k$-trees. 
Our strategy for the proof is, in fact, to prove the solvability of $k$-trees, 
by induction on a measure of complexity of $k$-trees.  
Unfortunately, for infinite domains, the solvability of $k$-trees is in general weaker than having a $(k,k+1)$-subinstance, and this brings several technical complications into our proof. 
In particular, we will be working with  $\csp$ instances,  that \textbf{won't necessarily be  $(k,k+1)$-instances}, or
even $k$-uniform.

The remaining parts of this section are organized as follows. 
In the first subsection we introduce concepts that are useful for working with instances and their solutions -- patterns and realizations. 
The next subsection studies solvability  with a fixed evaluation for $k$ variables 
and provides two core technical claims for the inductive proof of the solvability of $k$-trees; 
the proof is then assembled in  the third subsection and the missing parts of Theorems~\ref{thm:mainresult} and~\ref{thm:sensfull} are derived as a consequence.

Until Theorem~\ref{thm:core} in the last section we  fix 
\begin{itemize}
    \item an integer $k \geq 2$,
    \item a variety $\var V$ with a $(k+2)$-ary near unanimity term in case of Theorem~\ref{thm:mainresult} 
    or an algebra $\alg A$ with local near unanimity term operations of arity $k+2$
    in case of Theorem~\ref{thm:sensfull};
    \item an instance $\inst I = (V, \{\alg A_x \mid x \in V\}, \constr) $, 
    where $\constr = \{(S,R_S) \mid S \subseteq V, |S| \leq k\}$ and $R_S \leq \prod_{x \in S} \alg A_x$, such that, for any $S' \subseteq S$ with $|S| \leq k$, the projection of $R_S$ onto $S'$ is contained in $R_{S'}$~%
    (here either every $\alg A_x$ is in a variety $\var V$ in case of Theorem~\ref{thm:mainresult}, 
    or $\alg A_x$ is an isomorphic copy of $\alg A$ in case of Theorem~\ref{thm:sensfull}).
\end{itemize}

A $(k,k+1)$-instance can  naturally be expanded to meet the condition in the last item by adding the constraints $(S',R_{S'})$ for $|S'|<k$, where $R_{S'}$ is defined as the projection of  $R_S$ onto $S'$ for an arbitrary $k$-element superset $S$ of $S'$.
It is an easy exercise, and we leave it to the reader, to verify that this definition does not depend on the choice of $S$.

For a tuple of (not necessarily distinct) variables $x_1$, \dots, $x_l$ with $l \leq k$ we denote $R_{x_1, \dots, x_l} = \{(r_{x_1}, \dots, r_{x_l}) \mid \vc{r} \in R_{\{x_1, \dots, x_l\}}\} \leq \prod_{i=1}^l \alg A_{x_i}$.
Finally, we set $A = \bigcup_{x \in V} A_x$. 

\subsection{Patterns}

A pattern is a hypergraph whose  vertices are labeled by variables and hyperedges indicate that constraints should be satisfied. It will be convenient to have the set of hyperedges closed under taking subsets.

\begin{definition}
A \emph{pattern} is a triple $\pat{P} = (P; \asc{F},v)$, where $P$ is a set of \emph{vertices},  $\asc{F}$ is a family of at most $k$-element subsets of $P$ closed under taking subsets, and $v$ is a mapping $v: P \to V$.  Members of $\asc{F}$ are called \emph{faces} and the variable $v(i)$ is referred to as the \emph{label} of $i$.

A \emph{realization} of $\pat{P}$ is a mapping   $\alpha: P \to A$, which is \emph{consistent with} $v$, that is, $\alpha(i) \in A_{v(i)}$ for every $i \in P$, and \emph{satisfies} every face $\{f_1, \dots, f_l\} \in \asc{F}$, that is,  
$(\alpha(f_1), \dots, \alpha(f_l)) \in R_{v(f_1), \dots, v(f_l)}$.
\end{definition}

 For clarity, we will always call a mapping from a set of vertices to $A$ (which is not necessarily a realization of a pattern) an \emph{assignment}   (denoted $\alpha$, $\beta$, \dots), a mapping from a set of variables to $A$ an \emph{evaluation} (denoted $\phi$, $\rho$, \dots), and a mapping from a set of vertices to $V$ a \emph{labeling} (denoted $v$).
 We say that an assignment $\alpha$ \emph{extends} an evaluation $\phi$ if $\alpha(p) = \phi(v(p))$ for any $p$ in the domain of $\alpha$ such that $v(p)$ is in the domain of $\phi$.

 Since we assume that the $A_x$'s are disjoint, any assignment uniquely determines a consistent labeling and it makes sense to say that an assignment satisfies a set of vertices $F$, provided $|F| \leq k$. Also note that, by the assumptions on $\inst{I}$, if an assignment $\alpha$ satisfies $F$, then it satisfies every subset of $F$. 
 Finally, note that in the same situation $\alpha(i)=\alpha(i')$ whenever $v(i)=v(i')$. 

A pattern  $\pat{P}' = (P'; \asc{F}',v')$ is a \emph{subpattern} of $\pat{P}$ if $P' \subseteq P$,  $\asc{F}' \subseteq \asc{F}$, and $v'$ is the restriction of $v$ to $P'$.
By a union of two patterns we mean the set-theoretical union of the vertex sets, face sets, and labelings. It can only be formed if there are no collisions among labels (that is, each vertex belonging to both patterns must have the same label in both patterns). 

The richest patterns are the \emph{complete patterns}, whose faces are all the subsets of the vertex set of size at most $k$.
Note that a realization of a complete pattern with $l \leq k$ vertices is essentially the same as a tuple in the corresponding constraint relation.
The most important patterns for our purposes 
are $l$-trees with $l \leq k$. These are, informally, patterns obtained from the empty pattern by gradually adding complete patterns with at most $l+1$ vertices and merging them along a face to the already constructed pattern. 

\begin{definition}
Let $l \leq k$ and
let $F$ be a set of labeled vertices of size at most $l$.
The \emph{complete $l$-tree with base $F$ of depth 1} is the complete pattern with vertex set $F$. 
The \emph{complete $l$-tree with base $F$
of depth $d+1$} is obtained from the complete $l$-tree $\pat{P}$ with base $F$ of depth $d$ by adding to $\pat{P}$, for every face $E$ of $\pat{P}$ and every  $(l+1-|E|)$-element set of variables $U$,
a set $G$ of $|U|$ fresh vertices labeled by all elements of $U$ and all the at most $k$-element subsets of $E \cup G$ as faces. 

An \emph{$l$-tree} is a subpattern of a complete $l$-tree.
\end{definition}

The significance of $l$-trees is apparent from the following observation.

\begin{lemma} \label{lem:trees_real}
Assume that $\inst{I}$ is a $(k,k+1)$-instance~(with small arity constraints added).
Let $l \leq k$, let $\pat{P}$ be an $l$-tree, and let $F$ be a face of $\pat{P}$.
Then any assignment $\alpha: F \to A$ that satisfies $F$ can be extended to a realization of $\pat{P}$. 
In particular, every $l$-tree is realizable.
\end{lemma}
\begin{proof}
If $\pat{P}$ is a complete $l$-tree with base $F$, then $\alpha$ can be gradually extended to a realization of $\pat{P}$ by a straightforward application of the definition of $(k,k+1)$ instance.
It remains to observe that every $l$-tree with a face $F$ is a subpattern of a complete $l$-tree with base $F$.
\end{proof}

As noted above, realizability of $k$-trees in some sense even characterizes $(k,k+1)$-instances for finite domains. 
From this perspective it makes sense to use $k$-trees to measure the consistency level (called the quality) 
of a tuple in a constraint relation and, more generally, the consistency level of a realization. 

\begin{definition}
Let $F$ be a labeled set of vertices of size at most $k$. We say that an assignment $\alpha$, whose domain includes $F$ and which is consistent with the labeling, \emph{satisfies  $F$ with quality $d$} if $\restr{\alpha}{F}$ can be extended to a realization of the complete $k$-tree with base $F$ of depth $d$.
A realization $\alpha$ of a pattern $\pat{P}$ \emph{has quality} $d$ (or $\alpha$ satisfies $\pat{P}$ with quality $d$) if $\alpha$ satisfies each face of the pattern with quality $d$.

Similarly, we say that an evaluation $\phi: W \to A$ (where $|W| \leq k$) \emph{has quality} $d$ if
 the corresponding assignment for a $|W|$-element set of vertices labeled by all the elements of $W$ has quality $d$. 
\end{definition}

Informally, an evaluation $\phi$ has quality $d$ if it survives $d$ steps in a certain naturally defined consistency procedure. 
Note that a realization of a pattern is the same as a realization of quality $1$ and  a realization of quality $d$ is also a realization of quality $d'$ for any $d' \leq d$. 
Finally, observe that if an assignment $\alpha$ satisfies $F$ with quality $d$, then it satisfies every subset of $F$ with quality $d$.

We finish this subsection with two observations.

\begin{lemma} \label{lem:real_subalg}
The set of quality-$d$ realizations of a pattern $\pat{P}$ is a subuniverse of $\prod_{i \in P} \alg A_{v(i)}$.
\end{lemma}

\begin{proof}
For $d$=1 the claim is a straightforward consequence of the fact that constraint relations are subuniverses of products of $\alg A_x$'s. 
Otherwise we observe that the set of quality-$d$ realizations of $\pat{P}$ is the projection of the set of quality-1 realizations of a larger pattern $\pat{Q}$ to $P$. Indeed, $\pat{Q}$ can be taken as the pattern obtained from $\pat{P}$ by appending to every face $F$ the complete $k$-tree with base $F$ of depth $d$.
\end{proof}

\begin{lemma} \label{lem:face-to-tree}
Let $E \subseteq F$ be labeled sets of vertices, $E \leq k$, $|F| \leq k+1$,  
and let $\alpha:E \to A$ be an assignment which is consistent with the labeling and satisfies $E$ with quality $d+1$. Then $\alpha$ can be extended to an assignment $\beta: F \to A$ which is consistent with the labeling and satisfies each at most $k$ element subset of $F$ with quality $d$.

More generally, for any $k$-tree $\pat{P}$, any face $F$, and any $d$, there exists $d'$ such that  every assignment $\alpha:F \to A$ which satisfies $F$ with quality $d'$ can be extended to a realization of $\pat{P}$ of quality $d$.
\end{lemma}
\begin{proof}
The first observation follows from the definitions while the second one is proved by induction from the first one.
\end{proof}

\subsection{Fixing patterns}

A fixing pattern is a pattern together with a specified set $Y$ of fixing variables. The idea is to require that any consistent evaluation of $Y$ can be extended to a realization of the whole pattern. 
Since our instance isn't necessarily a $(k,k+1)$-instance the following modification is needed.

\begin{definition}
A \emph{fixing pattern} is a pair $(\pat{P},Y)$, where $\pat{P}$ is a pattern and  $Y$ is a set of variables of size at most $k$. The elements of $Y$ are called \emph{fixing variables}, the remaining variables from $v(P) \setminus Y$ are called \emph{inner}. 

A fixing pattern $(\pat{P}, Y)$ is \emph{f-realizable}  if for every $d$ 
 there exists $d' = z_{(\pat{P},Y)}(d) \geq d$ such that 
 every evaluation $\phi: Y \to A$ of quality $d'$
 can be extended to 
 a realization of $\pat{P}$ of quality $d$.
\end{definition}

It will be a feature of the proofs in this subsection that
the sufficient $d' = z_{(\pat{P},Y)}(d)$ from the definition will actually depend only on the ``shape'' of the fixing pattern: it will not depend on the instance, or on the variety, or on the concrete choice of labeling (i.e., the same $d'$ will work for a pattern obtained from $\pat{P}$ by changing $v$ to $r(v)$ and $Y$ to $r(Y)$, for any $r: V \to V$). 

A vertex $f$ of a fixing pattern $(\pat{P},Y)$ is called fixing/inner if the variable $v(f)$ is.
Faces consisting entirely of inner variables are called \emph{inner}, the remaining faces are called \emph{fixing}. 
A fixing face, whose set of inner vertices is $F$ and whose set of labels of fixing vertices is $Y'$, is denoted $[F,Y']$. Note that the definition of f-realization only depends on the ``inner part'' of the fixing pattern together with the list of those $[F,Y']$ that are present in the fixing pattern.
It will often be convenient to choose $\pat{P}$ \emph{free}, that is, the sets of fixing \textbf{vertices} of any two maximal fixing faces are disjoint.

An inner face $F$ is called \emph{completely fixed} if $[F,Y']$ is a (fixing) face for every $(k - |F|)$-element set of variables $Y' \subseteq Y$.
If $\pat{Q}$ is a pattern and $Y$ a set of variables of size at most $k$, which is disjoint from $v(Q)$,  
then the \emph{complete $Y$-fixing} (\emph{complete vertex $Y$-fixing}, respectively) of  $\pat{Q}$ is the free fixing pattern $(\pat{P},Y)$,
whose set of inner faces coincides with the set of faces of $\pat{Q}$ and each inner face (inner vertex, respectively) is completely fixed. 
Since complete fixings are chosen freely,  a complete fixing of a $k$-tree is a $k$-tree.

We say that a pattern $\pat{Q}$ is \emph{strongly realizable} if each complete fixing of $\pat{Q}$ is f-realizable.

Our aim, and the main technical contribution of this section is to prove that every $k$-tree is strongly realizable.  We now present, in Lemma~\ref{lem:sr1} and Lemma~\ref{lem:sr3}, two constructions that preserve f-realizability.
A proof that  the complete fixing of every $k$-tree  can be obtained by these constructions is contained in the next subsection.

\begin{lemma} \label{lem:sr1}
Let $1 \leq l \leq k+1$.
Let $(\pat{P},Y$)  be the complete vertex $Y$-fixing of a complete pattern $\pat{S}$ with $l$ vertices and, if $l \leq k-1$, freely add to $\pat{P}$ an additional fixing face $[S,Y']$ (and its subfaces) for some $Y' \subseteq Y$ of size $k-l$.

If each complete pattern with $l-1$ vertices is  strongly realizable,
then $(\pat{P},Y)$ is f-realizable.
\end{lemma}

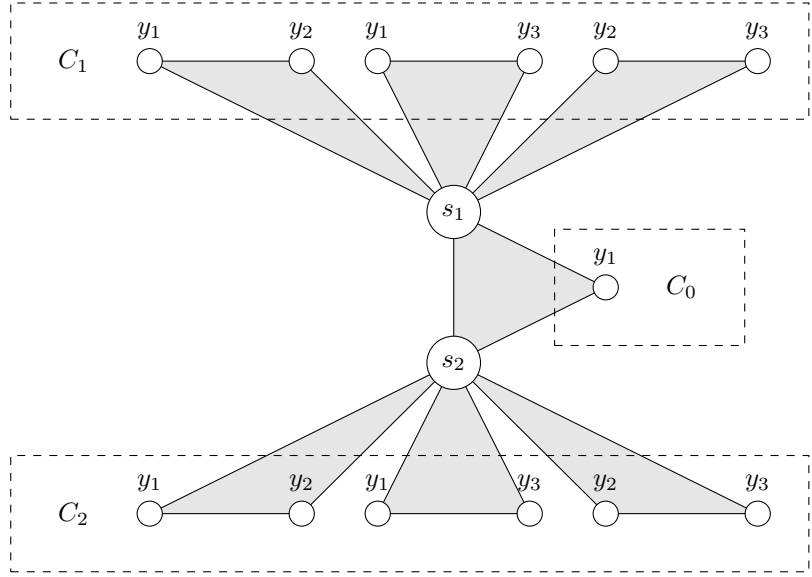
\begin{figure} 
\centering
\begin{tikzpicture}
    
    \filldraw[fill=gray!20, draw=black](1,0) -- (3,0) -- (5,-2) -- cycle;
    \filldraw[fill=gray!20, draw=black](4,0) -- (6,0) -- (5,-2) -- cycle;
    \filldraw[fill=gray!20, draw=black](7,0) -- (9,0) -- (5,-2) -- cycle;
    \filldraw[fill=gray!20, draw=black](5,-4) -- (7,-3) -- (5,-2) -- cycle;
    \filldraw[fill=gray!20, draw=black](1,-6) -- (3,-6) -- (5,-4) -- cycle;
    \filldraw[fill=gray!20, draw=black](4,-6) -- (6,-6) -- (5,-4) -- cycle;
    \filldraw[fill=gray!20, draw=black](7,-6) -- (9,-6) -- (5,-4) -- cycle;

    \node[draw=none] (1) at (0,0) {$C_1$};
    \node[shape=circle, draw=black, fill=white, label={$y_1$}] (2) at (1,0) {};
    \node[shape=circle, draw=black, fill=white, label={$y_2$}] (3) at (3,0) {};
    \node[shape=circle, draw=black, fill=white, label={$y_1$}] (4) at (4,0) {};
    \node[shape=circle, draw=black, fill=white, label={$y_3$}] (5) at (6,0) {};
    \node[shape=circle, draw=black, fill=white, label={$y_2$}] (6) at (7,0) {};
    \node[shape=circle, draw=black, fill=white, label={$y_3$}] (7) at (9,0) {};
    
    \node[fit=(1)(7), draw, dashed, inner sep=5mm] {};
    
    \node[shape=circle, draw=black, fill=white] (8) at (5,-2) {$s_1$};
    
    \node[shape=circle, draw=black, fill=white, label={$y_1$}] (9) at (7,-3) {};
    \node[draw=none] (10) at (8,-3) {$C_0$};
    \node[fit=(9)(10), draw, dashed, inner sep=5mm] {};
    
    \node[shape=circle, draw=black, fill=white] (11) at (5,-4) {$s_2$};

    \node[draw=none] (12) at (0,-6) {$C_2$};
    \node[shape=circle, draw=black, fill=white, label={$y_1$}] (13) at (1,-6) {};
    \node[shape=circle, draw=black, fill=white, label={$y_2$}] (14) at (3,-6) {};
    \node[shape=circle, draw=black, fill=white, label={$y_1$}] (15) at (4,-6) {};
    \node[shape=circle, draw=black, fill=white, label={$y_3$}] (16) at (6,-6) {};
    \node[shape=circle, draw=black, fill=white, label={$y_2$}] (17) at (7,-6) {};
    \node[shape=circle, draw=black, fill=white, label={$y_3$}] (18) at (9,-6) {};
    
    \node[fit=(12)(18), draw, dashed, inner sep=5mm] {};
    

\end{tikzpicture}

\caption{Case $k=3$, $l=2$ in the proof of Lemma~\ref{lem:sr1}.} \label{fig:Lemma_sr1}

\end{figure}

\begin{proof}
The case $l=1$ follows directly from Lemma~\ref{lem:face-to-tree} with the choice $d' = d+1$ and we henceforth assume $l>1$. 

Fix an arbitrary $d$. We need to choose $d'$ large enough so that the applications of the assumptions or Lemma~\ref{lem:face-to-tree}, which will be used in the proof, do not decrease the quality of our assignments below $d$. Specifically, we first choose $d''$ so that $d'' \geq d+2$ and, in case that $l=k+1$, also $d'' \geq z_{(\pat{Q},Z)}(d+1)$ for each complete fixing $(\pat{Q},Z)$ of a complete pattern with 2 vertices (note that $l = k+1 \geq 3$ in this case); and then choose $d'$ so that $d' \geq z_{(\pat{Q},Z)}(d'')$ for each complete fixing $(\pat{Q},Z)$ of a complete pattern with $l-1$ vertices (we will actually only use $(\pat{Q},Z)$ equal to $(\pat{P},Y)$ take away one inner vertex).

Denote $S = \{s_1, \dots, s_l\}$ the set of inner vertices of $(\pat{P},Y)$, $C_i$ (where $i \in [l]$) the set of fixing vertices coming from the vertex-fixing faces $[\{s_i\},\dots]$, and $C_0$ the set of fixing vertices coming from the fixing face  $[S,Y']$ (which is empty if $l \geq k$), see Figure~\ref{fig:Lemma_sr1}. Let $C = C_0 \cup C_1 \cup \dots \cup C_l$. Finally, let $\phi: Y \to A$ be an evaluation of quality $d'$.

We consider the set $T$ of restrictions of quality-$d$ realizations of $\pat{P}$ to the set $C$. Note that this set is a subuniverse of the product of the corresponding $\alg A_x$'s by Lemma~\ref{lem:real_subalg}.
\[
T = \{ \restr{\beta}{C}:
     \beta \mbox{ satisfies $\pat{P}$ with quality $d$}\} \leq \prod_{c \in C} \alg{A}_{v(c)}
\]
We need to prove that the tuple $\vc{a}$ defined by
$\vc{a}(c) = \phi(v(c))$ for all $c \in C$ is in $T$.
By the Baker-Pixley Theorem (Theorem~\ref{thm:BPold} when proving Theorem~\ref{thm:mainresult} and Theorem~\ref{thm:localBP} when proving Theorem~\ref{thm:sensfull}) it is enough to show that for any $(k+1)$-element set of coordinates $D$, the relation $T$ contains a tuple $\vc{b}$ that agrees with $\vc{a}$ on this set. This is now our aim. 

Denote $D_i = C_i \cap D$ and assume
 that there exists $i \geq 1$ such that $|D_0 \cup D_i| \leq k-l+1$. 
In this case we find a suitable tuple $\vc{b}$ in three steps as follows.
First, by the choice of $d'$,  we can extend $\phi$ to an assignment $\gamma: P \setminus \{s_i\} \to A$  that satisfies every $k$-element subset of $P \setminus \{s_i\}$ with quality $d''$, and set $\beta(p) = \gamma(p)$ for each $p \in P \setminus (\{s_i\} \cup C_0 \cup C_i)$.
Second, set $\beta(p) = \phi(v(p))$ for each $p \in D_0 \cup D_i$, let $F = (S \setminus \{s_i\}) \cup D_0 \cup D_i$, and note that $F$ has size at most $(l-1)+(k-l+1) = k$ and that $\beta$ satisfies $F$ with quality $d''$. Therefore, by Lemma~\ref{lem:face-to-tree}, $\restr{\beta}{F}$ can be extended to $F \cup \{s_i\}$ so that $\beta$ satisfies each at most $k$-element subset of $F \cup \{s_i\}$ with quality $d''-1 \geq d+1$. %
Third, for each face $E$ of $\pat{P}$ where $\beta$ is not yet fully defined we again use Lemma~\ref{lem:face-to-tree} and extend $\restr{\beta}{E \cap \dom(\beta)}$ to $E$ so that $\beta$ satisfies $E$ with quality $d$.
By construction, $\beta(c) = \phi(v(c))$ for every $c \in D$, and $\beta$ satisfies every face of $\pat{P}$ with quality $d$: the  fixing faces within $P \setminus (C_0 \cup C_i)$ because of the first step, the face $S$ because of the second step, and the remaining fixing faces (within $S \cup C_0 \cup C_i$) because of the third step. Therefore $\vc{b} = \restr{\beta}{C}$ is from $T$ and agrees with $\vc{a}$ on $D$.

Let $i \geq 1$ be such that $|D_i|$ is minimal. If $l \leq k$, then simple arithmetic gives us that $|D_0 \cup D_i| \leq k-l+1$ (so we are done in this case). Indeed, otherwise $|D_i| \geq k-l+2 - |D_0|$ and
$|D| \geq |D_0| + l|D_i| \geq |D_0| + l(k-l+2-|D_0|)$. For the maximum size of $D_0$, that is, $|D_0| = |C_0| =  k-l$, the right hand side of the last inequality is equal to $k+l$, and if $|D_0|$ decreases it gets bigger. Then $|D| \geq k+l > k+1$, a contradiction.

The remaining case is $l = k+1$ (in particular, $C_0 = D_0 = \emptyset$) and $|D_i| > k-l+1 = 0$ for each $1 \leq i \leq k+1$. Then, in fact, $D_i = \{d_i\}$ for each $i \geq 1$ (as $|D| \leq k+1$). By the pigeonhole principle, there are $i \neq j$ such that $v(d_i) = v(d_j)$. In this case we modify the three step procedure for finding $\vc{b}$ as follows. In the first step we define $\beta$ only on $P \setminus (\{s_i,s_j\} \cup C_i \cup C_{j})$,  in the second step we set $\beta(d_i) = \beta(d_j) = \phi(v(d_i))$, define $F = (S \setminus \{s_i,s_j\}) \cup D_i \cup D_{j}$, and instead of Lemma~\ref{lem:face-to-tree} we use the choice of $d''$ (coming from complete fixings of 2-element complete patterns) to extend $\restr{\beta}{F}$ to $F \cup \{s_i,s_j\}$. 
\end{proof}

The next lemma provides the base case for the second construction. 
We remark that having a near unanimity term of arity $2k$, when proving Theorem~\ref{thm:mainresult}, or local near unanimity term operations of arity $2k$, when proving Theorem~\ref{thm:sensfull},  is sufficient for the proof.

\begin{lemma} \label{lem:sr2}
Let $(\pat{P}_1,Y)$ and $(\pat{P}_2,Y)$ be free fixing patterns with exactly one common vertex $f$, 
which is labeled by $x \not\in Y$ and which is completely fixed in both patterns. For $i \in \{1,2\}$ let $\pat{P}'_i$ be the pattern obtained from $\pat{P}_i$ by removing the fixing vertices and all the vertices labeled $x$ (and all the incident faces).  Let $\pat{Q}$ be the union of $\pat{P}_1$ and $\pat{P}_2.$

If $(\pat{P}_i,Y)$, $i = 1,2$ are f-realizable and $\pat{P}'_i$, $i=1,2$  are strongly realizable, then $(\pat{Q},Y)$ is f-realizable. 
\end{lemma}

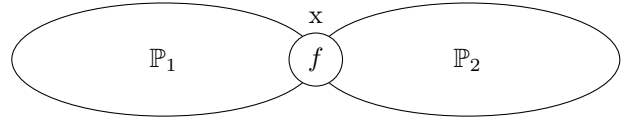
\begin{figure}
    \centering
    \begin{tikzpicture}
    
        \draw (-2,0) ellipse (2 and 0.75);
        \draw (2,0) ellipse (2 and 0.75);
    
        \node[draw=none] (P1) at (-2,0) {$\mathbb{P}_1$};
        \node[shape=circle, draw=black, fill=white,label={x}] (f) at (0,0) {$f$};
        \node[draw=none] (P2) at (2,0) {$\mathbb{P}_2$};
         
    \end{tikzpicture}
    \caption{Pattern $\mathbb{Q}$ in Lemma~\ref{lem:sr2}.}
    \label{fig:Lemma32}
\end{figure}

\begin{proof}
Fix $d$, choose $d''$ so that each complete fixing $(\pat{S},Z)$ of $\pat{P}'_1$ or $\pat{P}'_2$, which we will use in the proof, satisfies $d'' \geq z_{(\pat{S},Z)}(d+1)$, and  choose $d' \geq z_{(\pat{P}_i,Y)}(d'')$ for $i=1,2$. 

Let $\phi: Y \to A$ be an evaluation of quality $d'$ and denote $Y = \{y_1, \dots, y_k\}$ (where variables can possibly repeat).
For each $i \in \{1,2\}$ and $j \in \{1, \dots, k\}$ we construct a realization $\alpha_i^j: Q \to A$ of $\pat{Q}$ of quality $d$. The sought after quality-$d$ extension $\alpha$ of $\phi$ will be obtained by applying a $2k$-ary (local) near unanimity operation to these realizations. 
In order to construct $\alpha_i^j$  we first extend $\phi$ to a realization $\beta$ of $\pat{P}_i$ of quality $d''$ and define $\alpha_i^j(p) = \beta(p)$ for each $p \in \dom(\beta) = P_i$.
Next, we extend the evaluation $\rho: \{x\} \cup Y \setminus \{y_j\} \to A$, defined by $\rho(x) = \beta(f)$ and $\rho(y)=\phi(y)$ else, to a quality-$(d+1)$ realization $\gamma$ of the complete $(\{x\} \cup Y \setminus \{y_j\})$-fixing of $\pat{P}'_{3-i}$ and define $\alpha_i^j(c) = \gamma(c)$ for each $c \in \dom(\gamma)$ (noting that $\rho$ has quality $d''$ since $\beta$ does and $f$ is completely fixed in $\pat{P}_i$).
Finally, for each face $F$ of $\pat{Q}$ where $\alpha_i^j$ is not yet fully defined (this concerns fixing vertices of $\pat{P}_{3-i}$ labeled $y_j$) we use Lemma~\ref{lem:face-to-tree} and extend $\alpha_i^j$ so that it satisfies $F$ with quality $d$. Now $\alpha_i^j$ satisfies all the faces of $\pat{Q}$ with quality $d$ and agrees with $\phi$ on all of the fixing variables, except those from $\pat{P}_{3-i}$ labeled $y_j$. 
It follows that applying a $2k$-ary term operation to the $\alpha_i^j$ that satisfies the near unanimity condition for the set of components of the $\alpha_i^j$ gives an assignment of quality $d$ (by Lemma~\ref{lem:real_subalg}) that extends $\phi$, as required. 
\end{proof}

\begin{corollary} \label{cor:pppp}
Let $(\pat{P},Y)$ be a fixing pattern with two vertices $f_1 \neq f_2$ both labeled $x$ and both completely fixed, and let $n$ be a positive integer. Let $(\pat{Q},Y)$ be the fixing pattern obtained from the disjoint union of $n$ copies of $\pat{P}$ by identifying, for each $i \in \{1, \dots, n-1\}$, the vertex $f_2$ in the $i$-th copy with the vertex $f_1$ in the $(i+1)$-st copy. Let $\pat{P}'$ be the pattern obtained from $\pat{P}$ by removing the fixing vertices and all the vertices labeled $x$.

If $(\pat{P},Y)$ is f-realizable and $\pat{P}'$ is strongly realizable, then $(\pat{Q},Y)$ is f-realizable. 
\end{corollary}

\begin{figure}[ht]
    \centering
    \begin{tikzpicture}
    
        \draw (2,0) ellipse (2 and 0.75);
        \draw (6,0) ellipse (2 and 0.75);
        \draw (10,0) ellipse (2 and 0.75);
    
        \node[shape=circle, draw=black, fill=white,label={x}] (1) at (0,0) {$f_1$};
        \node[draw=none] (P1) at (2,0) {$\mathbb{P}$};
        \node[shape=circle, draw=black, fill=white,label={x}] (2) at (4,0) {$f_1'/f_2$};
        \node[draw=none] (P2) at (6,0) {$\mathbb{P}$};
        \node[shape=circle, draw=black, fill=white,label={x}] (3) at (8,0) {$f_1''/f_2'$};
        \node[draw=none] (P3) at (10,0) {$\mathbb{P}$};
        \node[shape=circle, draw=black, fill=white,label={x}] (4) at (12,0) {$f_2''$};
         
    \end{tikzpicture}
    \caption{Pattern $\mathbb{Q}$ in Corollary \ref{cor:pppp}.}
    \label{fig:Corollary_pppp}
\end{figure}
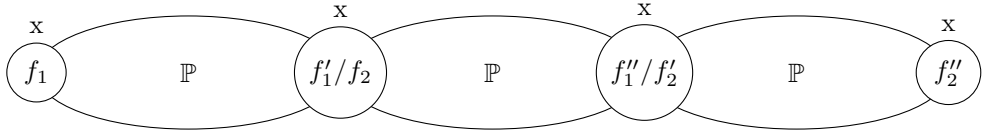

\begin{proof}
The proof follows by induction from Lemma~\ref{lem:sr2}, noting that in each step if we remove vertices labeled $x$ and fixing vertices from $\pat{Q}$, we get a pattern which is a disjoint union of strongly realizable patterns and is thus strongly realizable. 
\end{proof}

The following lemma provides the second construction. The proof uses Corollary~\ref{cor:pppp} (which requires a near unanimity term of arity $2k$ or local near unanimity term operations of arity $2k$) but the rest of the reasoning is based on the loop lemma stated in Theorem~\ref{thm:olsak_loop}, for which a near unanimity term (or local near unanimity term operations) of any arity is sufficient.

\begin{lemma} \label{lem:sr3}
Let $(\pat{P}_1,Y)$ and $(\pat{P}_2,Y)$ be fixing patterns with a common inner face $E$ and no other common vertices, such that both $\pat{P}_1$ and $\pat{P}_2$ are $k$-trees. For $i = 1,2$ let $f_i$  be a completely fixed inner vertex of $\pat{P}_i$ with label $x$ such that $E \cup \{f_i\}$ is a face of $\pat{P}_i$.
Let $\pat{Q}$ be the pattern obtained from the union of $\pat{P}_1$ and $\pat{P}_2$ by identifying vertices $f_1$ and $f_2$, and let $\pat{Q}'$ be the pattern obtained from $\pat{Q}$ (or $\pat{P}_1 \cup \pat{P}_2$) by removing the fixing vertices and all the vertices labeled $x$.

If $(\pat{P}_1 \cup \pat{P}_2,Y)$ is f-realizable and $\pat{Q}'$ is strongly realizable, then $(\pat{Q},Y)$ is f-realizable.
\end{lemma}

\begin{figure}
    \centering
    \begin{tikzpicture}
        
        \coordinate (f1) at (0,0);
        \coordinate (f2) at (4,-1);
        \coordinate (blank1) at (1.75,-1.25);
        \coordinate (blank2) at (1,-4);
        
        \filldraw[fill=gray!20, draw=black] (f1) -- (blank1) -- (blank2) -- cycle;
        \filldraw[fill=gray!20, draw=black] (f2) -- (blank1) -- (blank2) -- cycle;
        
        \node[shape=circle, draw=black,label={x}, fill=white] (f1node) at (f1) {$f_1$};
        \node[shape=circle, draw=black,inner sep=8pt, fill=white] (blank2node) at (blank2) {};
        \node[shape=circle, draw=black,inner sep=8pt, fill=white] (blank1node) at (blank1) {};
        \node[shape=circle, draw=black,label={x}, fill=white] (f2node) at (f2) {$f_2$};
        
        \path [-,line width=1pt, dashed] (f1node) edge[bend left] node [left] {} (f2node);

        \draw[black] 
            (f2node) 
            .. controls (7,-2) and (7,-4) .. 
            (blank2node);
            
        \draw[black] 
            (f1node) 
            .. controls (-3,0) and (-3,-4) .. 
            (blank2node);
        
        \node at (-1,-2) {$\mathbb{P}_1$};
        \node at (4,-2.5) {$\mathbb{P}_2$};
        
    \end{tikzpicture}
    \caption{Patterns $\pat{P}_1 \cup \pat{P}_2$ and $\pat{Q}$ in Lemma~\ref{lem:sr3}}
    \label{fig:Lemma33}
\end{figure}
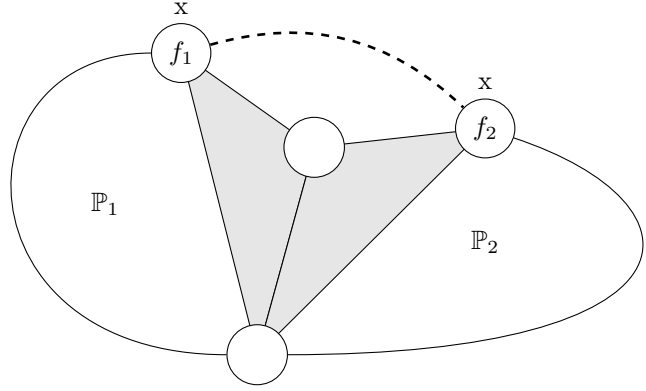

\begin{proof}
Let $r > 2$ be such that, in the case of proving Theorem~\ref{thm:mainresult}, $\var V$ has an $r$-ary near unanimity term, and in the case of proving Theorem~\ref{thm:sensfull}, $\alg A$ has local near unanimity term operations of arity $r$ (so $r = k+2$ works).  Let  $(\pat{Q}^{r-1},Y)$ be the fixing pattern obtained by taking
the disjoint union of $r-1$ copies of $\pat{P}_1 \cup \pat{P}_2$ and identifying the vertex $f_2$ in the $i$-th copy with the vertex $f_1$ in the $(i+1)$-first copy, for each $i \in \{1, \dots, r-2\}$.
The pattern $(\pat{Q}^{r-1},Y)$ is f-realizable by Corollary~\ref{cor:pppp}.

Fix $d$, choose $d''$ using Lemma~\ref{lem:face-to-tree} so that, for both $i \in \{1,2\}$, every quality-$d''$ assignment $\alpha: E \cup \{f_i\} \to A$  extends to a quality-$d$ realization of $\pat{P}_i$, and choose $d' \geq z_{(\pat{Q}^{r-1},Y)}(d''+1)$.

Let $\phi: Y \to A$ be an evaluation of quality $d'$.
Denote by $B$ the set of all elements of $a \in A_x$ such that the evaluation $x \mapsto a$ has quality $d''+1$, denote by $T$ the set of all the quality-$d$ realizations $\beta$  of  $\pat{P}_1 \cup \pat{P}_2$  such that both $\{f_1\}$ and $\{f_2\}$ have quality $d''+1$ and both $E \cup \{f_1\}$ and $E \cup \{f_2\}$ have quality $d''$, and denote by $S \subseteq T$  the set of those $\beta \in T$ that extend $\phi$. 
By a similar argument to that of Lemma~\ref{lem:real_subalg}, both $T$ and $S$ are subuniverses of $\prod_{p \in P_1 \cup P_2} \alg A_{v(p)}$.
Using the near unanimity term of arity $r$ (or local near unanimity term operations of arity $r$)
 $S$ clearly locally $r$-absorbs $T$.
The plan is to apply Theorem~\ref{thm:loop_inv} to the binary relation $\proj_{f_1,f_2} S \subseteq B \times B$. If this binary relation contains a loop, then the corresponding $\alpha \in S$ satisfies $\alpha(f_1)=\alpha(f_2)$ and, therefore, we actually obtain a realization of $\pat{Q}$ of quality $d$, as required.

It remains to verify the assumptions of Theorem~\ref{thm:loop_inv}.
By the choice of $d'$, the pattern $\pat{Q}^{r-1}$ has a quality-$(d''+1)$ realization that extends $\phi$. The images of copies of vertices $f_1$ and $f_2$ in such a realization yield a directed walk in $\proj_{f_1,f_2}(S)$ of length $r-1$. 
 Next, since $S$ locally $r$-absorbs $T$, then $\proj_{f_1,f_2}(S)$ locally $r$-absorbs $\proj_{f_1,f_2}(T)$, so it is enough to verify that the latter relation contains $=_B$ and $\proj_{f_1,f_2}(S)^{-1}$. For the first case, pick $b \in B$ and recall that the assignment $f_1 \mapsto b$ has quality $d''+1$ by the definition of $B$. We extend this assignment (using Lemma~\ref{lem:face-to-tree}) to a quality $d''$-assignment $\alpha: E \cup \{f_1\} \to A$, define $\alpha(f_2) = \alpha(f_1)$, and extend $\alpha$ to a quality-$d$ realization of $\pat{P}_1 \cup \pat{P}_2$. The obtained assignment witnesses $(b,b) \in \proj_{f_1,f_2}(T)$. Finally, to show that $\proj_{f_1,f_2}(T)$ contains $\proj_{f_1,f_2}(S)^{-1}$, consider any $(a,b) \in \proj_{f_1,f_2}(S)^{-1}$. By the definition of $S$, the pattern $\pat{P}_1 \cup \pat{P}_2$ has a realization $\alpha$ such that $\alpha(f_1)=b$, $\alpha(f_2)=a$, and both $E \cup \{f_1\}$ and $E \cup \{f_2\}$ have quality $d''$. We flip the values $\alpha(f_1)$ and $\alpha(f_2)$, restrict $\alpha$ to $E \cup \{f_1,f_2\}$ and extend this assignment using the choice of $d''$ to a  realization of $\pat{P}_1 \cup \pat{P}_2$ of quality $d$, giving us $(a,b) \in \proj_{f_1,f_2}(T)$ and concluding the proof.
\end{proof}

\subsection{Assembly}

Lemma~\ref{lem:sr1} and Lemma~\ref{lem:sr3} enable us to prove that every $k$-tree is strongly realizable. We split the inductive proof of this fact into two lemmata.

\begin{lemma}
Let $1 < l \leq k+1$ and assume, in case that $l>2$, that every $(l-2)$-tree is strongly realizable. Then every complete pattern with $l$ vertices is strongly realizable.
\end{lemma}

\begin{proof}
We start with a complete vertex $Y$-fixing of a complete pattern with $l$ vertices, which is f-realizable by Lemma~\ref{lem:sr1} (note that a complete pattern with $l-1$ vertices is an $(l-2)$-tree), and add fixing faces one by one while preserving the f-realizability. 

So, let  $\pat{S}$ be an $f$-realizable $Y$-fixing of a complete pattern with $l$ vertices and let $[E,Y']$ be such that $E = \{e_1, \dots, e_{l'}\}$, $l'
 < k$, is an inner face of $\pat{S}$  and $Y' \subseteq Y$ is a $(k-|E|)$-element set of variables. Our aim is to show that $\pat{S}$ plus the fixing face $[E,Y']$ is f-realizable. Let $(\pat{C},Y)$ be the complete vertex $Y$-fixing of a complete pattern with the set of inner vertices $G = \{g_1, \dots, g_{l'}\}$ (where $g_i$'s are fresh vertices) labeled according to $E$ (i.e.,  $v(g_i) = v(e_i)$ for each $i \in [l']$) with an additional fixing face $[G,Y']$. By Lemma~\ref{lem:sr1}, this fixing pattern is realizable.
Let $(\pat{C}^i, Y)$, $i \in \{0, \dots, l'\}$ be the fixing pattern obtained by renaming the vertices $g_1, \dots, g_i$ to $e_1, \dots, e_i$, respectively. The aim, reformulated, is to show that $(\pat{S} \cup \pat{C}^{i},Y)$ is f-realizable for $i=l'$. We prove this claim by induction on $i$. 

For $i=0$ the union $\pat{S} \cup \pat{C}^i$ is disjoint, therefore the claim follows from the f-realizability of $\pat{S}$ and $\pat{C}^0 = \pat{C}$. For the induction step from $i$ to $i+1$ we apply Lemma~\ref{lem:sr3} with $\pat{P}_1 = \pat{S}$, $\pat{P}_2 = \pat{C}^i$, $f_1=e_{i+1}$, and $f_2 = g_{i+1}$. Note that $(\pat{P}_1 \cup \pat{P}_2,Y)$ is f-realizable by the induction hypothesis and $\pat{Q}'$ is strongly realizable since it is an $(l-2)$-tree (or a single vertex in case that $l=2$), so we can conclude that $(\pat{Q},Y) = (\pat{S} \cup \pat{C}^{i+1},Y)$ is f-realizable, finishing the proof. 
\end{proof}

\begin{lemma}
Let $1 < l \leq k+1$ and assume that every complete pattern with $l$ vertices is strongly realizable. Then every $(l-1)$-tree is strongly realizable.
\end{lemma}

\begin{proof}
It is enough to show that every complete $(l-1)$-tree is strongly realizable. However, for an inductive proof of this claim, it will be convenient to use more general $(l-1)$-trees, those that can be obtained from the empty pattern in $n$ steps by taking the union of the already constructed pattern $\pat{S}$ with a complete pattern $\pat{C}$ on $l' \leq l$ vertices such that $S \cap C$ (where $0 \leq |S \cap C| < l'$) is a face in both patterns (with the same  labelling in both patterns). The induction is primarily on $n$ and secondarily on $|S \cap C|$.
For $n=1$ the claim follows from the assumption of the lemma. If $S \cap C = \emptyset$, then $\pat{S} \cup \pat{C}$ is a disjoint union and the claim follows by the inductive assumption and the assumption of the lemma.

Otherwise, take a fresh set $Y$ of $k$ variables and let $(\pat{Q},Y)$ be a complete $Y$-fixing of $\pat{S} \cup \pat{C}$. 
Pick a vertex in $S \cap C$, say vertex $f_1$ labeled $x$, let $\pat{C}'$ be the pattern obtained from $\pat{C}$ by renaming vertex $f_1$ to a fresh vertex $f_2$, let $(\pat{P}_1,Y)$ and $(\pat{P}_2,Y)$ be complete $Y$-fixings of $\pat{S}$ and $\pat{C}'$, respectively, and let $E = (S \cap C) \setminus \{f_1\}$. 
Note that this notation is consistent with the statement of Lemma~\ref{lem:sr3}: $\pat{Q}$ can be obtained from $\pat{P}_1 \cup \pat{P}_2$ by identifying vertices $f_1$ and $f_2$. 
To conclude the proof, we observe that the assumptions of Lemma~\ref{lem:sr3} are satisfied. Indeed,  $(\pat{P}_1 \cup \pat{P}_2,Y)$ is f-realizable by the inductive assumption (since it is a complete fixing of $\pat{S} \cup \pat{C}'$ for which $|S \cap C'| < |S \cap C|$) and
 $\pat{Q}'$ is strongly realizable since it is a subpattern of  $\pat{S} \cup \pat{C}'$.
\end{proof}

The following corollary is the core technical contribution of this section. Its proof follows by induction from the previous two lemmata.

\begin{corollary} \label{cor:trees_strong_real}
Every $k$-tree is strongly realizable. 
\end{corollary}

Armed with Corollary~\ref{cor:trees_strong_real},
we are ready to execute the idea outlined in the beginning of this section.
For the purpose of the following theorem, we call an instance  $\inst I = (V, \{\alg A_x \mid x \in V\}, \constr)$ a \emph{weak $k$-instance} 
if it satisfies the running assumption, that is,
$\constr = \{(S,R_S) \mid S \subseteq V, |S| \leq k\}$ and, for any $S' \subseteq S$ such that $|S| \leq k$, the projection of $R_S$ onto $S'$ is contained in $R_{S'}$.

\begin{theorem} \label{thm:core}
Let $k \geq 2$ and $n \geq 0$ be integers. Then there exists $d = z(n,k)$ such that for any variety $\var V$ with a $(k+2)$-ary  near unanimity term, or any idempotent algebra $\alg A$ with local near unanimity term operations of arity $k+2$,
any weak $k$-instance $\inst I$ of $\csp(\var V)$ (or $\csp(\alg A)$) with at most $n$ variables, and any at most $k$-element set of variables $Y$, every evaluation $\phi: Y \to A$ of quality $d$ extends to a solution of $\inst I$. 
\end{theorem}

\begin{proof}
We prove the claim by induction on $n$.  If $n \leq 1$, then the claim trivially holds with $d = 1$.  Otherwise, we denote $d' = z(n - 1,k)$ and  pick a $d$ greater than or equal to $z_{(\pat{T},Y)}(d')$ for every complete $Y$-fixing $(\pat{T},Y)$ of a complete $k$-tree  of depth $d'$ (formed over the variables $V \setminus Y$). 

Consider an instance $\inst{I}$ of $\csp(\var V)$ (or $\csp(\alg A)$) and an evaluation $\phi: Y \to A$ of quality $d$.
We define a new instance $\inst I' = (V', \{\alg A_x \mid x \in V'\}, \{(S,R'_S) \mid S \subseteq V, |S| \leq k\})$ by setting
$V' = V \setminus Y$ and 
$$
R'_S = \{\restr{\rho}{S} \mid \rho: Y \cup S \to A \mbox{ is a partial solution of $\inst{I}$ such that }
\restr{\rho}{Y} = \phi\}
$$
Clearly, $\inst I'$ is a weak $k$-instance.
We have chosen $d$ so that, in the instance $\inst I$, the partial evaluation $\phi$ extends to a realization of the complete $Y$-fixing of a complete $k$-tree of depth $d'$ over the set of variables $V \setminus Y$ (the base can be chosen arbitrarily for the argument). This realization witnesses that, in the instance $\inst I'$, there exists an evaluation 
of quality $d'$.
By the choice of $d'$, any such evaluation extends to a solution $\theta$ of $\inst{I}'$.
Now $\phi \cup \theta$ is a solution of $\inst{I}$, finishing the proof.
\end{proof}

To conclude, we state the parts of Theorem~\ref{thm:mainresult} and Theorem~\ref{thm:sensfull} that we set out to prove in this section as the following corollary.  It directly follows from Lemma~\ref{lem:trees_real} and the previous theorem.

\begin{corollary}
\begin{enumerate}
\item If $\var V$ is a variety that has a $(k+2)$-ary near unanimity term then every $(k,k+1)$-instance of the $\csp$ over $\var V$
is sensitive.
\item If $\alg A$ is an idempotent algebra that has local near unanimity term operations of arity $k+2$ then every $(k,k+1)$-instance of $\csp(\alg A)$ is sensitive.
\end{enumerate}
\end{corollary}

\end{document}